\documentclass{lmcs}
\pdfoutput=1

\usepackage{lastpage}
\lmcsdoi{17}{3}{24}
\lmcsheading{}{\pageref{LastPage}}{}{}%
{Jul.~30,~2019}{Sep.~17,~2021}{}

\usepackage{amsmath,amsthm,amssymb,stmaryrd,xspace,pdfsync,mathrsfs,verbatim,bbm}

\usepackage{microtype}
\usepackage{tikz}
\usepackage[mathscr]{euscript}
\usepackage[english]{babel}
\usetikzlibrary{decorations.text,arrows,hobby,backgrounds,decorations.pathreplacing,shapes,fit,shadows,calc,patterns,intersections}

\usepackage[utf8]{inputenc}
\usepackage[T1]{fontenc}

\definecolor{my1}{cmyk}{0,.6,0,0}
\definecolor{my2}{cmyk}{.3,.0,.0,.0}

\newcommand{\Gb}{\ensuremath{\mathbf{G}}\xspace}
\newcommand{\Hb}{\ensuremath{\mathbf{H}}\xspace}

\newcommand{\Kb}{\ensuremath{\mathbf{K}}\xspace}
\newcommand{\Lb}{\ensuremath{\mathbf{L}}\xspace}

\newcommand{\Qb}{\ensuremath{\mathbf{Q}}\xspace}

\newcommand{\Vb}{\ensuremath{\mathbf{V}}\xspace}

\newcommand{\nat}{\ensuremath{\mathbb{N}}\xspace}

\newcommand{\Cs}{\ensuremath{\mathcal{C}}\xspace}
\newcommand{\Ds}{\ensuremath{\mathcal{D}}\xspace}
\newcommand{\Gs}{\ensuremath{\mathcal{G}}\xspace}
\newcommand{\Is}{\ensuremath{\mathcal{I}}\xspace}

\newcommand{\Ps}{\ensuremath{\mathcal{P}}\xspace}

\newcommand{\Fs}{\ensuremath{\mathcal{F}}\xspace}

\newcommand{\Rs}{\ensuremath{\mathcal{R}}\xspace}

\newcommand{\quotienting}{quotient-closed\xspace}

\newcommand{\vari}{\quotienting Boolean algebra\xspace}

\newcommand{\pvari}{quotient-closed lattice\xspace}

\newcommand{\pvaris}{quotient-closed lattices\xspace}

\newcommand{\at}{\ensuremath{\textup{AT}}\xspace}

\newcommand{\bool}[1]{\ensuremath{{Bool\/}(#1)}\xspace}
\newcommand{\pol}[1]{\ensuremath{{Pol\/}(#1)}\xspace}
\newcommand{\bpol}[1]{\ensuremath{{BPol\/}(#1)}\xspace}

\newcommand{\copol}[1]{\ensuremath{{co\/}\textup{-}{\it Pol\/}(#1)}\xspace}

\newcommand{\cocl}[1]{\ensuremath{{co\/}\textup{-}#1}}

\newcommand{\bswd}{\ensuremath{\mathcal{B}\Sigma_{2}(<)}\xspace}
\newcommand{\bswsd}{\ensuremath{\mathcal{B}\Sigma_{2}(<,+1)}\xspace}

\newcommand{\imprint}{imprint\xspace}
\newcommand{\imprints}{imprints\xspace}

\newcommand{\Imprints}{Imprints\xspace}

\newcommand{\tame}{multiplicative\xspace}
\newcommand{\Ratms}{Rating maps\xspace}

\newcommand{\ratms}{rating maps\xspace}
\newcommand{\ratm}{rating map\xspace}
\newcommand{\Nice}{Nice\xspace}
\newcommand{\nice}{nice\xspace}

\newcommand{\mratm}{multiplicative rating map\xspace}
\newcommand{\mratms}{multiplicative rating maps\xspace}

\newcommand{\Mratms}{Multiplicative rating maps\xspace}

\newcommand{\prin}[2]{\ensuremath{\Is[#1](#2)}\xspace}

\newcommand{\opti}[2]{\ensuremath{\Is_{#1}\left[#2\right]}\xspace}

\newcommand{\popti}[3]{\ensuremath{\Ps_{#1}^{#2}[#3]}\xspace}

\newcommand{\pocopti}{\popti{\pol{\Cs}}{\alpha}{\rho}}

\newcommand{\bpopti}{\opti{\bpol{\Cs}}{\rho}}
\newcommand{\cbpopti}{\popti{\bpol{\Cs}}{\Cs}{\rho}}

\newcommand{\rbpol}[1]{\ensuremath{\Rs^\rho_{#1}}\xspace}
\newcommand{\rbpols}{\ensuremath{\rbpol{S}}\xspace}

\newcommand{\veps}{\ensuremath{\varepsilon}\xspace}

\newcommand{\typ}[2]{\ensuremath{[#1]_{#2}}\xspace}
\newcommand{\ctype}[1]{\typ{#1}{\Cs}}

\newcommand{\cmult}{\ensuremath{\mathbin{\scriptscriptstyle\bullet}}\xspace}
\newcommand{\upset}[2][\Cs]{\ensuremath{{\mathord\uparrow_{#1}\,#2}}}

\newcommand{\dclos}{\ensuremath{\mathord{\downarrow}}\xspace}

\newcommand{\dclosp}[1]{\ensuremath{\mathord{\downarrow_{#1}}}\xspace}
\newcommand{\dclosr}{\dclosp{R}}

\newcommand{\equc}{\ensuremath{\sim_\Cs}\xspace}

\newcommand{\canoc}{\ensuremath{\leqslant_\Cs}\xspace}
\newcommand{\canod}{\ensuremath{\leqslant_\Ds}\xspace}
\newcommand{\canog}{\ensuremath{\leqslant_\Gs}\xspace}
\newcommand{\inv}{\ensuremath{^{-1}}}

\newcommand{\sclac}{\ensuremath{{A^*}/{\sim_{\Cs}}}\xspace}

\tikzstyle{nor}=[minimum size=0.35cm,draw,rounded rectangle,inner sep=2pt]
\tikzstyle{nod}=[minimum size=0.35cm,draw,circle,inner sep=2pt]
\tikzstyle{nok}=[minimum size=0.45cm,draw,circle,inner sep=0pt]
\tikzstyle{nof}=[minimum size=0.35cm,draw,circle,double,double
distance=1pt]
\tikzstyle{port}=[minimum size=0.35cm,draw,thick,rectangle,inner sep=2pt]
\tikzstyle{nop}=[minimum size=0.35cm,draw,thick,rectangle,inner sep=1pt,rotate=90]

\tikzstyle{nol}=[minimum size=0.35cm,draw,rounded rectangle,inner sep=1pt,rotate=90]

\tikzstyle{ar}=[line width=0.5pt,->,double]
\tikzstyle{siar}=[line width=1.5pt,->]
\tikzstyle{ars}=[line width=1.5pt,->,double]

\newcommand{\brataux}[2]{\ensuremath{\xi_{#1}[#2]}\xspace}

\newcommand{\bratauxd}{\brataux{\Ds}{\rho}}
\newcommand{\bratauxbc}{\brataux{\bpol{\Cs}}{\rho}}
\newcommand{\bratauxbd}{\brataux{\bool{\Ds}}{\rho}}

\newcommand{\lrataux}[3]{\ensuremath{\zeta_{#1}^{#2}[#3]}\xspace}

\newcommand{\lratauxd}{\lrataux{\Ds}{\alpha}{\rho}}

\newcommand{\quasi}[1]{\ensuremath{\mu_{#1}}\xspace}
\newcommand{\quasir}{\quasi{\rho}}
\newcommand{\quasit}{\quasi{\tau}}

\theoremstyle{plain}
\newtheorem{theorem}[thm]{Theorem}
\newtheorem{corollary}[thm]{Corollary}
\newtheorem{proposition}[thm]{Proposition}
\newtheorem{lemma}[thm]{Lemma}
\newtheorem{fct}[thm]{Fact}
\newtheorem*{claim}{Claim}
\theoremstyle{definition}
\newtheorem{example}[thm]{Example}
\newtheorem{remark}[thm]{Remark}
\newtheorem*{notation}{Notation}

\title{Separation for dot-depth two}

\author{Thomas~Place}
\address{LaBRI, Universit\'e de Bordeaux, Institut Universitaire de France}

\author{Marc~Zeitoun}
\email{tplace@labri.fr}
\email{mz@labri.fr}

\keywords{Words, regular languages, concatenation hierarchies, first-order logic, quantifier alternation, membership, separation} 
\subjclass{F.4.1,F.4.3}

\thanks{Supported by the Agence Nationale de la Recherche, DeLTA project (ANR-16-CE40-0007).}

\begin{document}
\begin{abstract}
  The dot-depth hierarchy of Brzozowski and Cohen classifies the star-free languages of finite words. By a theorem of McNaughton and Papert, these are also the first-order definable languages. The dot-depth rose to prominence following the work of Thomas, who proved an exact correspondence with the quantifier alternation hierarchy of first-order logic: each level in the dot-depth hierarchy consists of all languages that can be defined with a prescribed number of quantifier blocks. One of the most famous open problems in automata theory is to settle whether the membership problem is decidable for each level: is it possible to decide whether an input regular language belongs to this level?

  Despite a significant research effort, membership by itself has only been solved for low levels. A recent breakthrough was achieved by replacing membership with a more general problem: \emph{separation}. Given two input languages, one has to decide whether there exists a third language in the investigated level containing the first language and disjoint from the second. The motivation for looking at separation is threefold: (1) while more difficult, it is more rewarding, as solving it requires a better understanding; (2) being more general, it provides a more convenient framework, and (3) all recent membership algorithms are actually reductions to separation for lower levels.

  We present a separation algorithm for dot-depth two. A key point is that while dot-depth two is our most prominent application, our theorem is more general. We consider a family of hierarchies, which includes the dot-depth: concatenation hierarchies. They are built through a generic construction process: one first chooses an initial class of languages, the basis, which serves as the lowest level in the hierarchy. Further levels are built by applying generic operations. Our main theorem states that for any concatenation hierarchy whose basis consists of finitely many languages, separation is decidable for level one. In the special case of the dot-depth, this can be lifted to level two using previously known results.
\end{abstract}

\maketitle

\section{Introduction}
\label{sec:intro}
\noindent\textbf{Concatenation hierarchies.}
Many fundamental problems about regular languages raised in the~70s~\cite{jep-openreg35} led to considerable advances, not only in automata theory but also in logic and algebra, thanks to the discovery of deep connections between these areas. Even if some of these questions are now well understood, a few others remain wide open, despite a wealth of research work spanning several decades. This is the case for the fascinating dot-depth problem~\cite{jep-dd45}, which has two elementary formulations: a language-theoretic one and a logical one. The language-theoretic one is the older of the two. It takes its roots in a theorem of Schützenberger~\cite{sfo} (see also~\cite{dkrh-textbook,dgfo}), which gives an algorithm to decide whether a regular language is star-free, \emph{i.e.}, can be expressed using union, complement and concatenation, but without the Kleene~star operator. This celebrated result was highly influential for three reasons:

\begin{itemize}
\item First, Schützenberger precisely formalized the objective of ``understanding the expressive power of a formalism'' through a decision problem called \emph{membership}, which asks whether an input regular language belongs to the class under~study.
\item Next, he developed a methodology for tackling it, which he applied to membership for the class of star-free languages.
\item Finally, McNaughton and Papert~\cite{mnpfo} established that star-free languages are exactly the first-order definable~ones.
\end{itemize}
This work highlighted the robustness of the notion of regularity, underlining the ties between automata theory and logic, and revealing new links with algebra. It also established membership as the reference problem for investigating classes of languages.

\medskip
Schützenberger's theorem led Brzozowski and Cohen to define the dot-depth hierarchy~\cite{BrzoDot}, an infinite classification~\cite{BroKnaStrict} of all star-free languages counting the number of alternations between concatenations and complements needed to define them. This definition is a particular instance of a generic construction process, which was formalized later and named \emph{concatenation hierarchies}. Any such hierarchy has a single parameter: a ``level $0$ class of languages'' (its \emph{basis}). Then, one uses two operations, polynomial and Boolean closure, to build two kinds of classes: half levels 1/2, 3/2, 5/2\dots and full levels 1, 2, 3\dots Given a class of languages~\Cs, its \emph{polynomial closure} \pol\Cs is the least class of languages containing~\Cs that is closed under union and marked concatenation ($K,L\mapsto KaL$, where $a$ is a letter). Its \emph{Boolean closure} \bool\Cs is the least class containing \Cs and closed under union and complement. For any full level $n$, the next half and full levels are built as follows:
\begin{itemize}
\item Level $n+\frac12$ is the polynomial closure of level~$n$.
\item Level $n+1$ is the Boolean closure of level $n+\frac12$.
\end{itemize}
Thus, a concatenation hierarchy is fully determined by its basis. Here, we are interested in hierarchies with a \emph{finite} basis, \emph{i.e.}, consisting of a finite number of regular languages.

\medskip
The most prominent hierarchies of this kind in the literature are the dot-depth itself, as well as the Straubing-Thérien hierarchy~\cite{StrauConcat,TheConcat}. They acquired this status when it was discovered~\cite{ThomEqu,PPOrder} that each of them coincides with the quantifier alternation hierarchy within an appropriate variant of first-order logic. These two variants have the same overall expressiveness but slightly different signatures (which impacts the properties that one can define at a given level of their quantifier alternation~hierarchies).

These correspondences motivated a research program to solve membership for all levels of both hierarchies, thus also characterizing the alternation hierarchies of first-order logic. However, progress has been slow. Until recently, the classes that were solved for both variants are only level 1/2~\cite{arfi87,pwdelta}, level~1~\cite{simon75,knast83} and level 3/2~\cite{arfi87,pwdelta,gssig2}. See~\cite{dgk-smallfragments} for a survey. Following these results, membership for level~2 remained open for a long time and was named the ``dot-depth 2~problem''.

\medskip\noindent\textbf{Separation.} Recently~\cite{pzqalt,pseps3j}, solutions were found for levels 2, 5/2 and 7/2. The key ingredient is a new problem stronger than membership: \emph{separation}. Rather than asking whether an input language belongs to the class \Cs under investigation, the \Cs-separation problem takes as input \emph{two} languages, and asks whether there exists a third one \emph{from \Cs} containing the first and disjoint from the second. While the interest in separation is recent, it has quickly replaced membership as the central question. A first practical reason is that separation proved itself to be a key ingredient in obtaining all recent membership results. See~\cite{PZ:Siglog15,pzgenconcat} for an overview. A striking example is provided by a crucial theorem of~\cite{pzqalt}. It establishes a generic reduction from \pol\Cs-separation to \Cs-membership which holds for any class~\Cs. Combined with a separation algorithm for level 3/2 and a little extra work, this yields a membership algorithm~for~level~5/2.

However, the main reason is deeper. The primary motivation for considering such problems is to thoroughly understand the classes under investigation. In this respect, while harder, separation is also far more rewarding than membership. On one hand, a membership algorithm for a class \Cs only applies to languages of \Cs: it can detect them and build a description witnessing membership. On the other hand, a separation algorithm for \Cs is universal: it applies to \textbf{\emph{any}} language. Indeed, one may view separation as an approximation problem: given an input pair $(L_{1},L_{2})$ one wants to over-approximate $L_{1}$ by a language in~\Cs, and $L_{2}$ serves to specify what a satisfying approximation is. This is why we look at separation: it yields a more robust understanding of the classes of languages than membership does.

The state of the art for separation is the following: it was shown to be decidable for levels 1/2, 1, 3/2 and 5/2 in the Straubing-Thérien hierarchy~\cite{martens,pvzmfcs13,pzqalt,pseps3j}. These results can be lifted to dot-depth using a generic transfer theorem~\cite{pzsucc2}. Notice the gap between levels 3/2 and 5/2: no algorithm is known for level~2. This is explained by the fact that obtaining separation algorithms presents very different challenges for half levels and for full levels. Indeed, it turns out that most separation algorithms rely heavily on closure under marked concatenation, which holds for half levels by definition, but not for full levels.

\medskip\noindent
\textbf{Contributions.} Our main result is a separation algorithm for level~2 in the Straubing-Thérien hierarchy. Furthermore, by the aforementioned transfer theorem~\cite{pzsucc2}, this can be lifted to separation for dot-depth~2. A crucial point is that this separation result is actually an instance of a generic theorem, which applies to any \emph{finite} class~\Cs satisfying a few standard properties (namely closure under Boolean operations and quotients). It states that for such a class \Cs, the class $\bool{\pol\Cs}$ has decidable separation. This has two important consequences:
\begin{itemize}
\item In \emph{any} hierarchy whose basis is such a class, level~1 has decidable separation.
\item In the specific case of the Straubing-Thérien hierarchy, this extends to level 2, since it is also level 1 of another finitely based concatenation hierarchy~\cite{pin-straubing:upper}.
\end{itemize}
This generic result complements other recent results in a natural way. It has been shown in~\cite{pseps3j} that \pol{\Cs}- and \pol{\bool{\pol\Cs}}-separation are decidable for any finite class~\Cs satisfying the aforementioned hypotheses. Combined with our results, this implies that for finitely based concatenation hierarchies, separation is decidable for all levels up to~3/2.

Our proof argument exploits a theorem of~\cite{pseps3j} about the simpler class \pol{\Cs}. However, the techniques that we use here are very different from the ones used in~\cite{pseps3j}, which rely heavily on the fact that \pol{\Cs} and \pol{\bool{\pol\Cs}} are closed under concatenation (again, this is not the case for \bool{\pol\Cs}).

Being generic, our approach yields separation algorithms for a whole family of classes. Moreover, it serves to pinpoint the key hypotheses which are critical in order to solve separation for dot-depth 2. Let us also stress that we obtain new direct proofs that separation is decidable for level 1 in both the dot-depth and Straubing-Thérien hierarchies. This is of particular interest for dot-depth 1 since the previous solution was indirect, as it relied on a transfer result from~\cite{pzsucc2}.

Finally, a key point is that all our results are stated and proved using a general framework that was recently introduced in~\cite{pzcovering2}. It is designed to handle the separation problem and to present the solutions in an elegant manner. In fact, this framework considers a third decision problem, which is even more general than separation: \emph{covering}. Given a class~\Cs, the \Cs-covering problem takes two objects as input, a single regular language $L$ and a finite set of regular languages \Lb. It asks whether there exists a \Cs-cover of $L$ (\emph{i.e.}, a finite set of languages in \Cs whose union includes $L$) such that no language in this \Cs-cover intersects all languages in \Lb. It is simple to show that separation is the special case of covering when the set \Lb is a singleton. Our main theorem actually states that \bool{\pol{\Cs}}-covering is decidable for any finite class \Cs satisfying mild hypotheses.

\medskip\noindent\textbf{Organization.} The paper is organized as follows. Section~\ref{sec:prelims} gives preliminary definitions. In Section~\ref{sec:concat}, we introduce concatenation hierarchies and state our generic theorem: \bool{\pol\Cs}-covering is decidable for every finite \vari \Cs. The remainder of the paper is devoted to presenting the algorithm. In Section~\ref{sec:covers}, we recall the framework which was designed in~\cite{pzcovering2} to handle the covering problem.  We use this framework in Section~\ref{sec:bpol} to formulate our algorithm for \bool{\pol{\Cs}}-covering. The remaining sections are devoted to proving that this algorithm is correct.

\medskip

This is the journal version of~\cite{pzboolpol}. The main result (\emph{i.e.}, that \bool{\pol\Cs}-separation is decidable for every finite \vari \Cs) has been generalized to \bool{\pol\Cs}-covering. Moreover, both its presentation and its proof have been completely reworked. Our results are now formulated using the general framework introduced in~\cite{pzcovering2}, designed to handle separation and covering problems.

\section{Preliminaries}
\label{sec:prelims}
In this preliminary section, we provide standard definitions for the basic objects investigated in the paper. Moreover, we present the separation and covering problems.

\subsection{Words, languages and classes}

For the whole paper, we fix an arbitrary finite alphabet $A$. Recall that $A^*$ denotes the set of all words over $A$, including the empty word~\veps. We let $A^{+}=A^{*}\setminus\{\veps\}$. If $u,v \in A^*$ are words, we write $u \cdot v$ or $uv$ the word obtained by concatenating $u$ and $v$.

A subset of $A^*$ is called a \emph{language}. We shall denote the singleton language $\{u\}$ by $u$. It is standard to generalize the concatenation operation to languages: given $K,L \subseteq A^*$, we write $KL$ for the language $KL = \{uv \mid u \in K \text{ and } v \in L\}$. Moreover, we shall also consider \emph{marked concatenation}, which is less standard. Given $K,L \subseteq A^*$, a marked concatenation of $K$ with $L$ is a language of the form $KaL$ for some letter $a \in A$.

A \emph{class of languages} is simply a set of languages. In the paper, we work with robust classes, \emph{i.e.}, which satisfy standard closure properties:
\begin{itemize}
\item A \emph{lattice} is a class closed under finite union and finite intersection, and containing the languages $\emptyset$ and $A^*$.
\item A \emph{Boolean algebra} is a lattice closed under complement.
\item Finally, a class is \emph{quotient-closed} when for every language $L$ of the class and every word $w \in A^*$, the following two languages also belong to the class:
  \[
    w^{-1}L \stackrel{\text{def}}= \{u \in A^* \mid wu \in L\} \quad \text{and} \quad Lw^{-1} \stackrel{\text{def}}= \{u \in A^* \mid uw \in L\}.
  \]
\end{itemize}
In this paper, we often use the letter ``\Ds'' to denote a lattice and the letter ``\Cs'' to denote a Boolean algebra (we also use \Cs to denote any generic class). The results of the paper apply to classes which are (at least) \pvaris. Moreover, they only apply to subclasses of the class of \emph{regular languages}. These are the languages that can be equivalently defined by monadic second-order logic, finite automata or finite monoids. Let us briefly recall the definition based on monoids, which we shall use.

A \emph{monoid} is a set $M$ equipped with an associative multiplication (usually denoted by ``$\cdot$'') which has a neutral element (usually denoted by ``$1_M$''). Recall that an \emph{idempotent} within a semigroup $S$ is an element $e \in S$ such that $ee = e$. Observe that $A^*$ is a monoid when equipped with word concatenation as the multiplication (the neutral element is \veps).  Hence, given a monoid $M$, we may consider morphisms $\alpha: A^* \rightarrow M$. Given such a morphism, we say that a language $L \subseteq A^*$ is \emph{recognized} by $\alpha$ when there exists $F \subseteq M$ such that $L = \alpha\inv(F)$. It is well-known that the regular languages are exactly those which can be recognized by a morphism $\alpha: A^* \rightarrow M$, where $M$ is a \textbf{finite} monoid.

\begin{remark}
  Whenever we consider a morphism $\alpha: A^* \to M$ in the paper, the monoid $M$ will be finite. For the sake of avoiding clutter, this will be implicit from now on.
\end{remark}

\subsection{Separation and covering}

In this paper, we investigate specific classes of languages which are part of concatenation hierarchies (introduced in Section~\ref{sec:concat}). We do so by relying on two decision problems: separation and covering. The former is standard while the latter was introduced in~\cite{pzcovering2}. Both of them are parameterized by an arbitrary class of languages~\Cs. Let us start with the definition of separation.

\medskip
\noindent
{\bf Separation.} Given three languages $K,L_1,L_2$, we say that $K$ \emph{separates} $L_1$ from $L_2$ if $L_1 \subseteq K \text{ and } L_2 \cap K = \emptyset$. Given a class of languages \Cs, we say that $L_1$ is \emph{\Cs-separable} from~$L_2$ if some language in \Cs separates $L_1$ from~$L_2$. Observe that when \Cs is not closed under complement, the definition is not symmetrical: $L_1$ could be \Cs-separable from $L_2$ while $L_2$ is not \Cs-separable from $L_1$. The separation problem associated to a given class \Cs is as~follows:

\medskip

\begin{tabular}{rl}
  {\bf INPUT:}  &  Two regular languages $L_1$ and $L_2$. \\
  {\bf OUTPUT:} &  Is $L_1$ \Cs-separable from $L_2$?
\end{tabular}

\medskip
\noindent
Separation is meant to be used as a mathematical tool in order to investigate classes of languages. Intuitively, obtaining a \Cs-separation algorithm requires a solid understanding of~\Cs.

\begin{remark}
  The \Cs-separation problem generalizes another classical decision problem: \Cs-membership, which asks whether a single regular language $L$ belongs to \Cs. Indeed, $L \in \Cs$ if and only if $L$ is \Cs-separable from $A^* \setminus L$.
\end{remark}

\noindent
{\bf Covering.} Our second problem is more general and was introduced in~\cite{pzcovering2}. Given a language $L$, a \emph{cover of $L$} is a \emph{\bf finite} set of languages \Kb such that:
\[
  L \subseteq \bigcup_{K \in \Kb} K.
\]
We speak of \emph{universal cover} to mean a cover of $A^*$. Moreover, if \Cs is a class, a \Cs-cover of $L$  is a cover \Kb of $L$ such that all $K \in \Kb$ belong to \Cs.

\medskip
Covering takes as input a language $L$ and a \emph{finite set of languages} $\Lb$. A cover \Kb of $L$ is \emph{separating} for $\Lb$ if for every $K\in\Kb$, there exists $L' \in \Lb$ that satisfies $K \cap L' = \emptyset$. In other words, no language of~$\Kb$ may intersect all languages of~$\Lb$.

Figure~\ref{fig:bgen:separatingcover} shows two covers of a language~$L$: $\{K_{1},K_{2}\}$ (on the left of the picture) and $\{K'_{1},K'_{2}\}$ (on the right). Let $\Lb=\{L_{1},L_{2}\}$. The cover  $\{K_{1},K_{2}\}$ of $L$ is separating for $\Lb$, since $K_{1}\cap L_{2}=\emptyset$ and $K_{2}\cap L_{1}=\emptyset$. On the other hand, the cover $\{K'_{1},K'_{2}\}$ is not separating for $\Lb$, as $K'_{1}$ intersects both $L_{1}$ and~$L_{2}$.

\begin{figure}[!ht]
  \begin{center}
    \begin{tikzpicture}
      \draw[thick] (0,0) circle (0.7cm) node (l1) {$L_1$};
      \draw[thick] (3,0) circle (0.7cm) node (l2) {$L_2$};
      \draw[thick] ($(0,0)!1!-60:(3,0)$) circle (0.7cm) node (l3) {$L$};

      \coordinate (m1) at ($(l1)+(1.0,0.0)$);
      \coordinate (m2) at ($(l2)+(1.0,0.0)$);
      \coordinate (m3) at ($(l3)+(1.0,0.0)$);

      \begin{pgfonlayer}{background}
        \def\fstlang{($(l3)$) to[closed,curve
          through={
            ($(l3)!1.5!160:(m3)$)
            .. ($(l3)!1.3!225:(m3)$)
            .. ($(l3)!0.8!290:(m3)$)
          }]
          ($(l3)$)}

        \def\seclang{($(l3)!0.8!-90:(m3)$) to[closed,curve
          through={
            ($(l3)!0.5!180:(m3)$)
            .. ($(l1)!0.3!90:(m1)$)
            .. ($(l2)!0.3!90:(m2)$)
            .. ($(l2)!1!0:(m2)$)
            .. ($(l3)!1!0:(m3)$)
          }]
          ($(l3)!0.8!-90:(m3)$)}

        \begin{scope}[draw opacity=0.7,fill opacity=0.3]
          \draw[thick,red] \fstlang;
          \draw[thick,blue] \seclang;
          \fill[red] \fstlang;
          \fill[blue] \seclang;
        \end{scope}

        \node at ($(l3)-(1.0,0.0)$) {$K'_2$};
        \node at ($(l3)!1/2!30:(l2)$) {$K'_1$};
      \end{pgfonlayer}

      \begin{scope}[xshift=-7.5cm]

        \draw[thick] (0,0) circle (0.7cm) node (l1) {$L_1$};
        \draw[thick] (3,0) circle (0.7cm) node (l2) {$L_2$};
        \draw[thick] ($(0,0)!1!-60:(3,0)$) circle (0.7cm) node (l3) {$L$};

        \coordinate (m1) at ($(l1)+(1.0,0.0)$);
        \coordinate (m2) at ($(l2)+(1.0,0.0)$);
        \coordinate (m3) at ($(l3)+(1.0,0.0)$);

        \begin{pgfonlayer}{background}
          \def\fstlang{($(l1)!1.3!160:(m1)$) to[closed,curve
            through={
              ($(l1)!1!90:(m1)$)
              .. ($(l2)!1!90:(m2)$)
              .. ($(l2)!1.3!20:(m2)$)
              .. ($(l2)!1/2!(l1)$)
            }]
            ($(l1)!1.3!160:(m1)$)}
          \def\seclang{($(l3)!0.8!280:(m3)$) to[closed,curve
            through={
              ($(l3)!0.8!210:(m3)$)
              .. ($(l1)!0.7!170:(m1)$)
              .. ($(l1)!0.3!90:(m1)$)
              .. ($(l1)!0.7!10:(m1)$)
              .. ($(l3)!1!90:(m3)$)
            }]
            ($(l3)!0.8!280:(m3)$)}

          \def\trdlang{($(l3)!0.8!-100:(m3)$) to[closed,curve
            through={
              ($(l3)!0.8!-30:(m3)$)
              .. ($(l2)!0.7!10:(m2)$)
              .. ($(l2)!0.3!90:(m2)$)
              .. ($(l2)!0.7!170:(m2)$)
              .. ($(l3)!1!90:(m3)$)
            }]
            ($(l3)!0.8!-100:(m3)$)}

          \begin{scope}[draw opacity=0.7,fill opacity=0.3]
            \draw[thick,blue] \seclang;
            \draw[thick,green] \trdlang;
            \fill[blue] \seclang;
            \fill[green] \trdlang;
          \end{scope}

          \node at ($($(l3)!1/2!(l1)$)-(0.5,0.0)$) {$K_1$};
          \node at ($($(l2)!1/2!(l3)$)+(0.5,0.0)$) {$K_2$};
        \end{pgfonlayer}

\end{scope}
    \end{tikzpicture}
  \end{center}
  \caption{Two covers of $L$. The left one is separating for $\{L_1,L_2\}$ and the right one is not.}
  \label{fig:bgen:separatingcover}
\end{figure}

Finally, given a class \Cs, we say that the pair $(L_,\Lb)$ is \Cs-coverable when there exists a \Cs-cover of $L$ which is separating for $\Lb$.

\medskip\noindent
The \Cs-covering problem is now defined as follows:

\medskip

\begin{tabular}{rl}
  {\bf INPUT:}  &  A regular language $L$ and a finite set of regular languages $\Lb$.\\
  {\bf OUTPUT:} &  Is $(L,\Lb)$ \Cs-coverable?
\end{tabular}

\medskip

Note that in the covering problem, we are interested in the existence of covers which are \emph{separating} (but notice that we do not keep this precision in the name of the problem itself, to lighten the terminology). It is straightforward to prove that covering generalizes separation provided that the class is a lattice, as stated in the following lemma (see Theorem~3.5 in~\cite{pzcovering2} for the proof).

\begin{lemma}\label{lem:septocove}
  Let \Ds be a lattice and let $L_1,L_2$ be two languages. Then $L_1$ is \Ds-separable from~$L_2$ if and only if $(L_1,\{L_2\})$ is \Ds-coverable.
\end{lemma}

In the paper, we shall not work with covering directly. Instead, we use a framework introduced in~\cite{pzcovering2}, which is designed to formulate and handle these problems in a more convenient manner. We recall this framework in Section~\ref{sec:covers}.

\subsection{Finite lattices}

In the paper, we work with classes built from an arbitrary finite lattice (\emph{i.e.}, one that contains finitely many languages) using generic two operations: polynomial closure and Boolean closure (see Section~\ref{sec:concat}). Let us present standard results about such~classes.

\medskip
\noindent
{\bf Canonical preorder relations.} Consider a finite lattice \Ds. It is classical to associate a \emph{canonical preorder relation over $A^*$} to \Ds. Given $w,w' \in A^*$, we write $w \canod w'$ if and only if the following holds:
\[
  \text{For all $L \in \Ds$,} \quad w \in L \ \Rightarrow\ w' \in L.
\]
It is immediate from the definition that \canod is transitive and reflexive, making it a preorder.

\begin{example}\label{ex:cowatpre}
  Consider the class \Ds consisting of all unions of intersections of languages of the form $A^*aA^*$, for $a\in A$. By definition, this class is a finite lattice. The associated preorder $\leqslant_\Ds$ is defined by $w \leqslant_\Ds w'$ if and only if every letter occurring in $w$ also occurs in~$w'$.
\end{example}
\noindent
We shall use several results about the relation \canod. We omit the proofs, which are simple and available in~\cite{pseps3j}.

\medskip
The first lemma is where we use the hypothesis that \Ds is finite. We say that a language $L \subseteq A^*$ is an \emph{upper set} (for \canod) when for any two words $u,v \in A^*$, if $u \in L$ and $u \canod v$, then $v \in L$. Furthermore, given $u \in A^*$, we let $\upset[\Ds]{u} \subseteq A^*$ be the least upper set containing~$u$: $\upset[\Ds]{u} = \{v \in A^* \mid u \canod v\}$.

\begin{lemma}\label{lem:canosatur}
  Let $\Ds$ be a finite lattice. Then, for any $L \subseteq A^*$, we have $L \in \Ds$ if and only if~$L$ is an upper set for \canod. In particular, \canod has finitely many upper sets.
\end{lemma}

Let us complete these definitions with a few additional useful results. First, as we observed for \at in Example~\ref{ex:ateq}, when the finite lattice is actually a Boolean algebra, it turns out that its canonical preorder is an equivalence relation. If \Cs is such a Boolean algebra, we shall denote this equivalence relation by \equc (instead of \canoc).

\begin{lemma}\label{lem:canoequiv}
  Let $\Cs$ be a finite Boolean algebra. Then, for any alphabet $A$, the canonical preorder \canoc is an equivalence relation \equc, which admits the following direct definition:
  \[
    \text{$w \equc w'$ if and only if for all $L \in \Cs$,\quad $w \in L \ \Leftrightarrow\ w' \in L$}.
  \]
  Thus, for any $L \subseteq A^*$, we have $L \in \Cs$ if and only if $L$ is a union of \equc-classes. In particular, \equc has finite index.
\end{lemma}

\begin{example}\label{ex:ateq}
  Consider the class \at of all Boolean combinations of languages $A^*aA^*$, for $a \in A$ (``\at'' stands for ``alphabet testable'': $L \in \at$ if and only if membership of a word $w$ in $L$ depends only on the letters occurring in $w$). Clearly, \at is a finite Boolean algebra---it is the Boolean closure of the class from Example~\ref{ex:cowatpre} (see Section~\ref{sec:closure-operations} for the definition of Boolean closure). In that case, $\leqslant_\at$ is an equivalence relation $\sim_\at$: $w \sim_\at w'$ if and only if $w$ and $w'$ have the same alphabet (\emph{i.e.}, contain the same set of letters).
\end{example}

Another important and useful property is that when \Ds is \quotienting, the canonical preorder \canod is compatible with word concatenation.

\begin{lemma}\label{lem:canoquo}
  A finite lattice \Ds is \quotienting if and only if its associated canonical preorder \canod is compatible with word concatenation. That is, for any words $u,v,u',v'$,
  \[
    u \canod u' \quad \text{and} \quad v \canod v' \quad \Rightarrow \quad uv \canod u'v'.
  \]
\end{lemma}

\medskip
\noindent
{\bf \Cs-compatible morphisms.} We use these notions to define special monoid morphisms. We fix a finite \vari \Cs for the definition.

\smallskip
By Lemma~\ref{lem:canoequiv}, \equc has finite index and the languages in \Cs are exactly the unions of \equc-classes. Moreover, since \Cs is \quotienting, we know from Lemma~\ref{lem:canoquo} that \equc is a congruence for word concatenation. It follows that the quotient set ${A^*}/{\equc}$ is a finite monoid and the map $w \mapsto \ctype{w}$ is a morphism from $A^*$ to ${A^*}/{\equc}$. Consider an arbitrary morphism $\alpha: A^* \to M$. We say that $\alpha$ is \emph{\Cs-compatible} when, for every $s \in M$, there exists a \equc-class denoted by \ctype{s} such that $\alpha\inv(s) \subseteq \ctype{s}$. In other words, $\ctype{w} = \ctype{s}$ for every $w \in A^*$ such that $\alpha(w) = s$.

\begin{remark}
  Let $\alpha:A^{*}\to M$ be a \Cs-compatible morphism. Given $s \in M$, the $\sim_{\Cs}$-class \ctype{s} is determined by $\alpha$ when $\alpha\inv(s) \neq \emptyset$ ($\ctype{s} = \ctype{w}$ for any $w \in \alpha\inv(s)$). If $\alpha\inv(s) = \emptyset$, we may choose any \equc-class as \ctype{s}. When we consider a \Cs-compatible morphism, we implicitly assume that the map $s \mapsto \ctype{s}$ is fixed.
\end{remark}

We prove that we may assume without loss of generality that all the morphisms we work with are \Cs-compatible: any regular language is recognized by a \Cs-compatible morphism.

\begin{lemma}\label{lem:compat}
  Let \Cs be a finite \vari. Given a regular language $L \subseteq A^*$, one can compute a \Cs-compatible morphism $\alpha: A^* \to M$ recognizing $L$.
\end{lemma}

\begin{proof}
  Since $L$ is regular, we can compute a finite monoid $N$ and a morphism $\beta: A^* \to N$ (not necessarily \Cs-compatible) recognizing $L$. Since we know that the quotient set ${A^*}/{\equc}$ is a finite monoid, the Cartesian product $M = N \times ({A^*}/{\equc})$ is a finite monoid for the componentwise multiplication. It now suffices to define the morphism $\alpha: A^* \to M$ by $\alpha(w) = (\beta(w),\ctype{w})$ for any $w \in A^*$. Clearly, $\alpha$ is a morphism which recognizes $L$ and one can verify that it is \Cs-compatible.
\end{proof}

\section{Closure operations and main theorem}
\label{sec:concat}
In this section, we define the family of classes that we investigate in the paper and present a few results about them. Then, we state our main theorem and discuss its consequences.

\subsection{Closure operations}
\label{sec:closure-operations}

Consider a class \Cs. The \emph{Boolean closure} of \Cs, denoted by \bool{\Cs}, is defined as the least Boolean algebra containing \Cs. The next lemma follows from the definitions (this amounts to verifying that quotients commute with Boolean operations, \emph{e.g.}, that for all languages $K,L\subseteq A^*$ and for all words $w$, we have $w^{-1}(K\cup L) = w^{-1}K\cup w^{-1}L $ and $w^{-1}(A^*\setminus L) = A^*\setminus(w^{-1} L)$, and symmetrically for the other quotient operation).

\begin{lemma}\label{lem:boolclos}
  Let \Ds be a \pvari. Then, \bool{\Ds} is a \vari.
\end{lemma}

The second operation that we shall consider is slightly more involved. Given a class~\Cs,  the \emph{polynomial closure} of \Cs, denoted by \pol{\Cs}, is the least class containing \Cs which is closed under both union and marked concatenation:
\[
  \text{for all $K,L \in \pol{\Cs}$ and $a \in A$}, \quad K \cup L \in \pol{\Cs} \text{ and } KaL \in \pol{\Cs}.
\]
While this is not obvious from the definition, when the input class \Cs is a \pvari, its polynomial closure \pol{\Cs} is a \pvari as well. The difficulty is to establish that \pol{\Cs} is closed under intersection. This was originally proved by Arfi~\cite{arfi87}, assuming additionally that \Cs is closed under complement. The result was then extended to \pvaris by Pin~\cite{jep-intersectPOL} (see also~\cite{pzgenconcat}).

\begin{theorem}[Arfi~\cite{arfi87}, Pin~\cite{jep-intersectPOL}]\label{thm:polclos}
  Let \Ds be a \pvari. Then, \pol{\Ds} is a \pvari closed under concatenation and marked concatenation.
\end{theorem}

In the paper, we consider classes of the form \bool{\pol{\Cs}} built by applying polynomial closure and Boolean closure successively to some arbitrary \vari~\Cs. For the sake of avoiding clutter, we shall write \bpol{\Cs} for \bool{\pol{\Cs}}. Note that by Lemma~\ref{lem:boolclos} and Theorem~\ref{thm:polclos}, we have the following corollary.

\begin{corollary}\label{cor:bpolclos}
Let \Ds be a \pvari. Then, \bpol{\Ds} is a \vari.
\end{corollary}

A key remark is that in general, classes built with Boolean closure (such as \bpol{\Cs}) are \emph{not} closed under concatenation. This contrasts with polynomial closure in which closure under (marked) concatenation holds by definition. This is an issue, as most of our techniques designed for handling separation and covering rely heavily on concatenation. We cope with this problem by using the following weak concatenation principle, which holds for any class that is the Boolean closure of another class which is itself closed under concatenation.

\begin{lemma}\label{lem:hintro:boolconcat}
  Let \Ds be a lattice closed under concatenation. Consider $L,L' \in \Ds$ and let $\Kb,\Kb'$ be \bool{\Ds}-covers of $L$ and $L'$ respectively. Then, there exists a \bool{\Ds}-cover~\Hb of~$LL'$ such that for every $H \in \Hb$, we have $H \subseteq KK'$ for some $K \in \Kb$ and $K' \in \Kb'$.
\end{lemma}

\begin{proof}
  Every language in $\Kb \cup \Kb'$ is a Boolean combination of languages in \Ds and $L,L' \in \Ds$. Therefore, there exists a finite lattice $\Gs \subseteq \Ds$ satisfying the two following properties:
  \begin{enumerate}
  \item $L,L' \in \Gs$ and,
  \item every language $K \in \Kb \cup \Kb'$ belongs to $\bool{\Gs}$.
  \end{enumerate}
  We define \Fs as the least lattice such that $HH' \in \Fs$ for every $H,H' \in \Gs$. Since  \Gs is finite, so is~\Fs. Moreover, we know that $\Gs \subseteq \Ds$ and \Ds is a lattice closed under concatenation. Therefore, $\Fs \subseteq \Ds$. It follows that \bool{\Fs} is a finite Boolean algebra such that $\bool{\Fs} \subseteq \bool{\Ds}$.

  Consider the canonical equivalence $\sim_{\bool{\Fs}}$ associated to \bool{\Fs} (it relates words that belong to the same languages of \bool{\Fs}, see Section~\ref{sec:prelims}). Since $L,L'\in \Gs$, it is immediate by definition of \Fs that we have $LL' \in \Fs \subseteq \bool{\Fs}$. Therefore, Lemma~\ref{lem:canoequiv} implies that $LL'$ is a union of $\sim_{\bool{\Fs}}$-classes. We let \Hb be the set consisting of all these $\sim_{\bool{\Fs}}$-classes. By definition, \Hb is a \bool{\Fs}-cover of $LL'$ and therefore a \bool{\Ds}-cover as well since $\bool{\Fs} \subseteq \bool{\Ds}$. It remains to prove that for every $H \in \Hb$, we have $H \subseteq KK'$ for some $K \in \Kb$ and $K' \in \Kb'$. We fix $H \in \Hb$ for the proof. We use the following fact.

  \begin{fct}\label{fct:hintro:boolconcat}
    Consider a \emph{finite} language $G \subseteq H$. Then, there exists $K \in \Kb$ and $K' \in \Kb'$ such that $G \subseteq KK'$.
  \end{fct}

  Let us first admit Fact~\ref{fct:hintro:boolconcat} and apply it to finish the main proof. For every $n \in \nat$, we let $G_n \subseteq H$ be the finite language containing all words of length at most $n$ in $H$. Clearly, we have,
  \[
    H = \bigcup_{n \in \nat} G_n \quad \text{and} \quad G_n \subseteq G_{n+1} \text{ for all $n \in \nat$}.
  \]
  For every $n \in \nat$, Fact~\ref{fct:hintro:boolconcat} yields $K_n \in \Kb$ and $K_n' \in \Kb'$ such that $G_n \subseteq K_nK_n'$. Since \Kb and $\Kb'$ are finite sets, there exist $K \in \Kb$ and $K' \in \Kb'$ such that $K_n = K$ and $K'_{n}=K'$ for infinitely many $n$. Since $G_n \subseteq G_{n+1}$ for all $n$, it follows that $G_n \subseteq KK'$ for every $n \in \nat$. Finally, since $H = \bigcup_{n \in \nat} G_n$, this implies $H \subseteq KK'$, finishing the proof.

  \medskip

  It remains to prove Fact~\ref{fct:hintro:boolconcat}. Consider a finite language $G \subseteq H$ and let $G = \{w_1,\dots,w_n\}$. We exhibit $K \in \Kb$ and $K' \in \Kb'$ such that $G \subseteq KK'$.

  By definition, $H$ is a $\sim_{\bool{\Fs}}$-class included in $LL'$. This implies that $w_1,\dots,w_n \in LL'$ and $w_1 \sim_{\bool{\Fs}} \cdots \sim_{\bool{\Fs}} w_n$. Using these equivalences, we first prove the following claim which involves the canonical preorder \canog associated to the finite lattice \Gs.

  \begin{claim}
    For every $u,v \in A^*$ such that $w_n = uv$, there exist $u_1,\dots,u_n,v_1,\dots,v_n \in A^*$ such that $w_i = u_iv_i$ for every $i \leq n$, $u \canog u_1 \canog \cdots \canog u_n$ and $v \canog v_1 \canog \cdots \canog v_n$.
  \end{claim}

  \begin{proof}
    We prove the existence of $u_1,v_1 \in A^*$ such that $w_1 = u_1v_1$, $u \canog u_1$ and $v \canog v_1$ using the hypothesis that $w_n = uv$ and $w_n \sim_{\bool{\Fs}} w_1$. One may then iterate the argument to build $u_2,\dots,u_n \in A^*$ and $v_2,\dots,v_n \in A^*$, using the fact that $w_1 \sim_{\bool{\Fs}} \cdots \sim_{\bool{\Fs}} w_n$.

    Consider the languages $U = \upset[\Gs] u$ and $V = \upset[\Gs] v$ (the upper sets of $u$ and $v$ for $\canog$). By Lemma~\ref{lem:canosatur}, we have $U,V \in \Gs$. Therefore, we have $UV \in \Fs$ by definition of \Fs. Clearly, $w_n = uv \in UV$. Therefore, $w_n \sim_{\bool{\Fs}} w_1$ implies that $w_1 \in UV$. This yields a decomposition $w_1 = u_1v_1$ with $u_1 \in U$ and $v_1 \in V$. By definition of $U,V$, this implies $u \canog u_1$ and $v \canog v_1$, finishing the proof.
  \end{proof}

  Since $w_n \in LL'$, it admits at least one decomposition $w_n = uv$ with $u \in L$ and $v \in L'$. Therefore, we may apply the claim: there exist $u_{1,1},\dots,u_{n,1},v_{1,1},\dots,v_{n,1} \in A^*$ such that $w_i = u_{i,1}v_{i,1}$ for every $i \leq n$, $u \canog u_{1,1} \canog \cdots \canog u_{n,1}$ and $v \canog v_{1,1} \canog \cdots \canog v_{n,1}$. Since $w_n=u_{n,1}v_{n,1}$, one can repeat the argument any number of times, say $m$, to obtain words $u_{i,j}, v_{i,j}$ for $1\leqslant i\leqslant n$ and $1\leqslant j\leqslant m$, such that for all such $i,j$, we have $w_i =u_{i,j} v_{i,j}$, and:
\begin{alignat*}{2}
  u &\canog u_{1,1} \canog \cdots \canog u_{n,1}  \qquad\text{and}\qquad &v &\canog v_{1,1} \canog \cdots \canog v_{n,1},\\
  u_{n,1} &\canog u_{1,2} \canog \cdots \canog u_{n,2}  \qquad\text{and}\qquad &v_{n,1} &\canog v_{1,2} \canog \cdots \canog v_{n,2},\\
  &\ \cdots\qquad&&\\
  u_{n,m-1} &\canog u_{1,m} \canog \cdots \canog u_{n,m}  \quad\ \,\text{and}\qquad& v_{n,m-1} &\canog v_{1,m} \canog \cdots \canog v_{n,m}.
\end{alignat*}

Moreover, since $w_1$ is a finite word, it admits finitely many decompositions $w_1 = u_1v_1$ with $u_1,v_1 \in A^*$. Therefore, the pigeonhole principle yields two integers $j<k$ such that $u_{1,j} = u_{1,k}$ (and therefore $v_{1,j} = v_{1,k}$), which gives words $u_1,\dots,u_n,v_1,\dots,v_n \in A^*$ such that $w_i = u_iv_i$ for every $i \leq n$ and,
  \[
    u \canog u_1 \canog \cdots \canog u_n \canog u_1 \qquad v \canog v_1 \canog \cdots \canog v_n \canog v_1.
  \]
  One may verify from the definitions that for every $x,y \in A^*$, $x \canog y$ and $y \canog x$ imply that $x \sim_{\bool{\Gs}} y$. Therefore, the above implies that,
  \begin{equation}
    u_1 \sim_{\bool{\Gs}} \cdots \sim_{\bool{\Gs}} u_n \qquad v_1 \sim_{\bool{\Gs}} \cdots \sim_{\bool{\Gs}} v_n.\label{eq:1}
  \end{equation}
  Since $u \in L$, $v \in L'$ and $L,L' \in \Gs$ by definition of \Gs, we get that $u_1 \in L$ and $v_1 \in L'$ by definition of $\canog$. Therefore, since \Kb and $\Kb'$ are covers of $L$ and $L'$ respectively, there exist $K \in \Kb$ and $K' \in \Kb'$ such that $u_1 \in K$ and $v_1 \in K'$.

  By definition of \Gs, both $K$ and $K'$ belong to $\bool{\Gs}$. It follows from~\eqref{eq:1} that $u_1,\dots,u_n \in K$ and $v_1,\dots,v_n \in K'$. Altogether, we obtain that $G = \{u_1v_1,\dots,u_nv_n\} \subseteq KK'$, finishing the proof.
\end{proof}

Finally, we shall need the following variant of Lemma~\ref{lem:hintro:boolconcat}, which considers marked concatenation instead of standard concatenation. The proof is identical to the one of Lemma~\ref{lem:hintro:boolconcat} and left to the reader.

\begin{lemma}\label{lem:hintro:boolmconcat}
  Let \Ds be a lattice closed under marked concatenation. Consider $L,L' \in \Ds$ and $a \in A$, and let $\Kb,\Kb'$ be \bool{\Ds}-covers of $L$ and $L'$ respectively. There exists a \bool{\Ds}-cover \Hb of $LaL'$ such that for every $H \in \Hb$, we have $H \subseteq KaK'$ for some $K \in \Kb$ and $K' \in \Kb'$.
\end{lemma}

\subsection{Main theorem}

We may now state the main theorem of the paper: whenever \Cs is a \emph{finite} \vari, \bpol{\Cs}-separation and \bpol{\Cs}-covering are both decidable.

\begin{theorem}\label{thm:bool}
  Let \Cs be a finite \vari. Then, separation and covering are decidable for \bpol{\Cs}.
\end{theorem}

Before we detail the applications of Theorem~\ref{thm:bool}, let us make an important observation. This result completes an earlier one presented in~\cite{pseps3j} which applies to classes of the form \pol{\Cs} and \pol{\bpol{\Cs}} (when \Cs is a finite \vari). Let us recall this result.

\begin{thmC}[\cite{pseps3j}]\label{thm:pbp}
  Let \Cs be a finite \vari. Then, separation and covering are decidable for \pol{\Cs} and \pol{\bpol{\Cs}}.
\end{thmC}

An important point is that while \bpol{\Cs} is an intermediary class between \pol{\Cs} and \pol{\bpol{\Cs}}, the proof of Theorem~\ref{thm:bool} involves ideas which are very different from those used in~\cite{pseps3j} to prove Theorem~\ref{thm:pbp}. This is not surprising and we already mentioned the reason above: unlike \pol{\Cs} and \pol{\bpol{\Cs}}, the class \bpol{\Cs} is not closed under concatenation in general. This difference is significant, since most of the techniques we have for handling covering rely heavily on concatenation. In practice, this means that Boolean closure is harder to handle than polynomial closure, at least with such techniques.

However, we do reuse a result of~\cite{pseps3j} to prove Theorem~\ref{thm:bool}. More precisely, it turns out that the arguments for handling \bpol{\Cs} (in this paper) and \pol{\bpol{\Cs}} (in~\cite{pseps3j}) both exploit the same subresult for the simpler class \pol{\Cs} (albeit in very different ways). This subresult is stronger than the decidability of \pol{\Cs}-covering and is proved in~\cite{pseps3j}. We recall it in Section~\ref{sec:pol}. However, it is important to keep in mind that from this preliminary result for \pol{\Cs}, the arguments for \bpol{\Cs} and \pol{\bpol{\Cs}} build in orthogonal directions.

\begin{remark}
  This discussion might seem surprising. Indeed, by definition, \pol{\bpol{\Cs}} is built from \bpol{\Cs} using polynomial closure. Hence, intuition suggests that one needs some knowledge about the latter to handle the former. However, this is not the case as Boolean closure can be bypassed in the definition of \pol{\bpol{\Cs}}. Specifically, one may prove that $\pol{\bpol{\Cs}} = \pol{\copol{\Cs}}$ where \copol{\Cs} is the class containing all complements of languages in \pol{\Cs}. This is exactly how the class $\pol{\bpol{\Cs}}$ is handled in~\cite{pseps3j}.
\end{remark}

The remaining sections of the paper are devoted to proving Theorem~\ref{thm:bool}. We rely on a framework which was designed in~\cite{pzcovering2} for the specific purpose of handling the covering problem. We recall it in Section~\ref{sec:covers}. The algorithm for \bpol{\Cs}-covering is presented in Section~\ref{sec:bpol}. The remaining sections are then devoted to the correction proof of this algorithm. However, let us first conclude the current section by detailing the important applications of Theorem~\ref{thm:bool}.

\subsection{Applications of the main theorem}

Classes of the form \bpol{\Cs} are important: they are involved in natural hierarchies of classes of languages, called \emph{concatenation hierarchies}. Let us briefly recall what they are (we refer the reader to~\cite{pzgenconcat} for a detailed presentation). A particular concatenation hierarchy depends on a single parameter: an arbitrary \vari of regular languages~\Cs, called its \emph{basis}. Once the basis is chosen, the construction is uniform. Languages are classified into levels of two kinds: full levels (denoted by 0, 1, 2,\ldots) and half levels (denoted by 1/2, 3/2, 5/2,\ldots):
\begin{itemize}
\item Level $0$ is the basis (\emph{i.e.}, our parameter class \Cs).
\item Each \emph{half level} $n+\frac{1}{2}$, for $n\in\nat$, is the \emph{polynomial closure} of the previous full level, \emph{i.e.}, of level $n$.
\item Each \emph{full level} $n+1$, for $n\in\nat$, is the \emph{Boolean closure} of the previous half level, \emph{i.e.}, of level $n+\frac12$.
\end{itemize}
The generic process is depicted in the following figure.
\begin{center}
  \begin{tikzpicture}[scale=.9]
    \node[anchor=east] (l00) at (0.0,0.0) {{\large $0$}};

    \node[anchor=east] (l12) at (1.5,0.0) {\large $\frac{1}{2}$};
    \node[anchor=east] (l11) at (3.0,0.0) {\large $1$};
    \node[anchor=east] (l32) at (4.5,0.0) {\large $\frac{3}{2}$};
    \node[anchor=east] (l22) at (6.0,0.0) {\large $2$};
    \node[anchor=east] (l52) at (7.5,0.0) {\large $\frac{5}{2}$};

    \draw[very thick,->] (l00) to node[above] {$Pol$} (l12);
    \draw[very thick,->] (l12) to node[below] {$Bool$} (l11);
    \draw[very thick,->] (l11) to node[above] {$Pol$} (l32);
    \draw[very thick,->] (l32) to node[below] {$Bool$} (l22);
    \draw[very thick,->] (l22) to node[above] {$Pol$} (l52);

    \draw[very thick,dotted] (l52) to ($(l52)+(1.0,0.0)$);
  \end{tikzpicture}
\end{center}

Hence, a reformulation of Theorem~\ref{thm:bool} is that for any concatenation hierarchy whose basis is \emph{finite}, separation is decidable for level one. There are two prominent examples of finitely based hierarchies:
\begin{itemize}
\item The \emph{Straubing-Th\'erien} hierarchy~\cite{StrauConcat,TheConcat}, whose basis is the class $\{\emptyset,A^*\}$.
\item The \emph{dot-depth} hierarchy of Brzozowski and Cohen~\cite{BrzoDot}, whose basis is the class $\{\emptyset,\{\varepsilon\},A^+,A^*\}$.
\end{itemize}
Consequently, Theorem~\ref{thm:bool} implies that separation and covering are decidable for level one in these two hierarchies. These are not new results. By definition, level one in the Straubing-Th\'erien hierarchy is exactly the class of \emph{piecewise testable languages}. For this class, separation has been solved in~\cite{martens,pvzmfcs13} and covering has been solved in~\cite{pzcovering2}. For dot-depth one, the decidability of covering and separation was originally obtained indirectly. Indeed, it is known~\cite{pzsucc2} that separation and covering for any level in the dot-depth hierarchy reduce to the same problem for the corresponding level in the Straubing-Th\'erien hierarchy. Therefore, while the decidability of separation and covering for dot-depth one is not a new result, an advantage of Theorem~\ref{thm:bool} is that we obtain a new \emph{direct} proof of this~result.

\medskip

However, these are not the main applications of Theorem~\ref{thm:bool}. It turns out that the theorem also applies to level two in the Straubing-Th\'erien hierarchy. Indeed, it is known that this level is also level one in another finitely based concatenation hierarchy. More precisely, recall the class \at presented in Example~\ref{ex:ateq}: it consists of all Boolean combinations of languages $A^*aA^*$, for some $a \in A$. It is straightforward to verify that \at is a finite \vari. The following theorem was shown in~\cite{pin-straubing:upper} (see also~\cite{pzcovering2} for a recent proof).

\begin{theorem}[\cite{pin-straubing:upper}]\label{thm:alphatrick}
  Level two in the Straubing-Th\'erien hierarchy is exactly the class \bpol{\at}.
\end{theorem}

In view of Theorem~\ref{thm:alphatrick}, we obtain the following immediate corollary of Theorem~\ref{thm:bool}.

\begin{corollary}\label{cor:st2}
  Separation and covering are decidable for level two in the Straubing-Th\'erien hierarchy.
\end{corollary}

Finally, this result may be lifted to dot-depth two using again the generic transfer theorem proved in~\cite{pzsucc2}. Hence, we obtain the following additional corollary.

\begin{corollary}\label{cor:dd2}
  Separation and covering are decidable for dot-depth two.
\end{corollary}

\begin{remark}
  Logical characterizations of these two hierarchies are known (see~\cite{ThomEqu} for the dot-depth hierarchy and~\cite{PPOrder} for the Straubing-Thérien hierarchy). Each of them corresponds to a quantifier alternation hierarchy within a particular variant of first-order logic over words. The two variants differ by the set of predicates which are allowed in sentences (they have the same overall expressive power, but this changes the levels in their respective quantifier alternation hierarchies). We refer the reader to~\cite{pzgenconcat} for details and a recent proof of these results.

  In particular, level two in the Straubing-Thérien hierarchy corresponds to a logic denoted by \bswd and dot-depth two corresponds to another logic denoted by \bswsd. Hence, our results also imply that covering and separation are decidable   for these two logics.
\end{remark}

\section{Framework: \ratms and optimal covers}
\label{sec:covers}
In this section, we present the framework which we use to formulate our covering algorithm for \bpol{\Cs} (when \Cs is a finite \vari) announced in Theorem~\ref{thm:bool}. The framework itself was designed and applied to several specific classes in~\cite{pzcovering2}. Moreover, it was also used in~\cite{pseps3j} to formulate algorithms for \pol{\Cs}- and \pol{\bpol{\Cs}}-covering. Here, we recall the part of this framework that we shall actually need in the paper. We refer the reader to~\cite{pzcovering2} for a complete and detailed presentation.

\medskip

Let \Ds be a lattice. In \Ds-covering, the input is a pair $(L,\Lb)$ where $L$ is a regular language and \Lb a finite set of regular languages: we have to decide whether there exists a \Ds-cover of~$L$ which is separating for \Lb. We design a two-step approach for tackling this~problem.

\medskip
The first step is to transform this decision problem into the following \emph{computational} problem. Given an input $(L,\Lb)$, we compute a \Ds-cover of $L$, which is \emph{optimal} with respect to~$\Lb$ in the following sense: for every subset $\Hb$ of $\Lb$, this optimal cover is separating for $\Hb$ if and only if $(L,\Hb)$ is \Ds-coverable. Thus, this single object encapsulates \emph{all answers} to the \Ds-covering problem for inputs of the form $(L,\Hb)$, with $\Hb\subseteq\Lb$. Such an optimal cover always exists when \Ds is a lattice. The set \Lb can be viewed as a parameter constraining this optimal cover, or dually, as a measure of the quality of this cover.

\medskip

The second step in the approach of~\cite{pzcovering2} consists in replacing the set \Lb by a (more general) algebraic object called a \emph{\ratm}. Intuitively, \ratms play the same role as the set of languages \Lb from the above paragraph. Thus, they are also designed to measure the quality of \Ds-covers. Given a \ratm $\rho$ and a language $L$, we use $\rho$ to rank the existing~\Ds-covers of $L$.

\medskip

We are then able to reformulate \Ds-covering with these notions. Instead of deciding whether $(L,\Lb)$ is \Ds-coverable, we compute an optimal \Ds-cover for $L$ for a \ratm $\rho$ that we build from $\Lb$. An advantage of this approach is that it yields elegant formulations for the covering algorithms which are formulated with it. We refer the reader to~\cite{pzcovering2} and~\cite{pseps3j} for examples (in addition to \bpol{\Ds}, which is presented in this paper). Another important motivation for using this framework is that in order to handle \bpol{\Ds}, we require a result for the simpler class \pol{\Ds} which is stronger than the decidability of covering (this result is proved in~\cite{pseps3j}). The framework of~\cite{pzcovering2} is designed to formulate this result.

\medskip

We start by defining \ratms. Then, we explain how they are used to measure the quality of a cover and define optimal covers. Finally, we connect these notions to the covering problem. Let us point out that several statements presented here are without proof. We refer the reader to~\cite{pzcovering2} for these proofs.

\subsection{\Ratms}

\Ratms involve commutative and idempotent monoids. We shall write such monoids $(R,+)$: we call the binary operation  ``$+$'' \emph{addition} and denote the neutral element by $0_R$. Being idempotent means that for all $r \in R$, we have $r + r = r$. Observe that for every commutative and idempotent monoid $(R,+)$, we may define a canonical ordering $\leq$ over $R$:
\[\text{For all }
  r, s\in R,\quad r\leq s \text{ when } r+s=s.
\]
It is straightforward to verify that $\leq$ is a partial order which is compatible with addition. Moreover, we have the following fact which is immediate from the definitions.

\begin{fct}\label{fct:increas}
  Let $(R,+)$ and $(Q,+)$ be two commutative and idempotent monoids. Moreover, let $\gamma: (R,+) \to (Q,+)$ be a morphism. Then, $\gamma$ is increasing: for every $s,t\in R$ such that $s \leq t$, we have $\gamma(s) \leq \gamma(t)$.
\end{fct}

\begin{example}
  For any set $E$, it is immediate that $(2^E,\cup)$ is an idempotent and commutative monoid. The neutral element is $\emptyset$. Moreover, the canonical ordering is set inclusion.
\end{example}

When manipulating the subsets of a commutative and idempotent monoid $(R,+)$ we shall often need to apply a \emph{downset operation}. Given $S \subseteq R$, we write $\dclosr S$ for the set,
\[
  \dclosr S = \{r \in R \mid r \leq s \text{ for some $s \in S$}\}.
\]
We extend this notation to Cartesian products of arbitrary sets with $R$. Given some set $X$ and a subset $S \subseteq X \times R$, we write $\dclosr S$ for the set,
\[
  \dclosr S = \{(x,r) \in X \times R \mid \text{there exists $q\in R$ such that $r \leq q$ and $(x,q) \in S$}\}.
\]

We may now define \ratms. A \emph{\ratm} is a morphism $\rho: (2^{A^*},\cup) \to (R,+)$ where $(R,+)$ is a \emph{finite} idempotent and commutative monoid called the \emph{rating set of $\rho$}. That is, $\rho$ is a map from $2^{A^{*}}$ to $R$ satisfying the following properties:
\begin{enumerate}\item\label{itm:fzer} $\rho(\emptyset) = 0_R$.
\item\label{itm:ford} For all $K_1,K_2 \subseteq A^*$, $\rho(K_1\cup K_2)=\rho(K_1)+\rho(K_2)$.
\end{enumerate}

For the sake of improved readability, when applying a \ratm $\rho$ to a singleton set $K = \{w\}$, we shall write $\rho(w)$ for $\rho(\{w\})$. Additionally, we write $\rho_*: A^* \to R$ for the restriction of $\rho$ to $A^*$: for every $w \in A^*$, we have $\rho_*(w) = \rho(w)$ (this notation allows us to write $\rho_*\inv(r) \subseteq A^*$ for the language of all words $w \in A^*$ such that $\rho(w) = r$).

\smallskip

Most of the statements involved in our framework make sense for arbitrary \ratms. However, we shall often have to work with special \ratms which satisfy additional properties. We present them now.

\medskip
\noindent
{\bf \Nice \ratms.} A \ratm $\rho: 2^{A^*} \to R$ is \emph{\nice} when, for every language $K \subseteq A^*$, there exist finitely many words $w_1,\dots,w_n \in K$ such that $\rho(K) = \rho(w_1) + \cdots + \rho(w_k)$.

Observe that in this case, $\rho$ is characterized by the canonical map $\rho_*: A^* \to R$. Indeed, for every language $K$, we may consider the sum of all elements $\rho(w)$ for $w \in K$: while it may be infinite, it boils down to a finite one since $R$ is commutative and idempotent. The hypothesis that $\rho$ is \nice implies that $\rho(K)$ is equal to this sum.

\medskip
\noindent
{\bf \Mratms.} A \ratm $\rho: 2^{A^*} \to R$ is \emph{\tame} when its rating set $R$ has more structure: it needs to be an \emph{idempotent semiring}. Moreover, $\rho$ has to satisfy an additional property connecting this structure to language concatenation. Namely, it has to be a morphism of semirings.

A \emph{semiring} is a tuple $(R,+,\cdot)$ where $R$ is a set and ``$+$'' and ``$\cdot$''  are two binary operations called addition and multiplication, such that the following axioms are satisfied:
\begin{itemize}
\item $(R,+)$ is a commutative monoid (the neutral element is denoted by $0_R$).
\item $(R,\cdot)$ is a monoid (the neutral element is denoted by $1_R$).
\item Multiplication distributes over addition: $r \cdot (s + t) = (r \cdot s) + (r \cdot t)$  and $(r + s) \cdot t = (r \cdot t) + (s \cdot t)$ for all $r,s,t \in R$.
\item The neutral element of $(R,+)$ is a zero for $(R,\cdot)$: $0_R \cdot r = r \cdot 0_R = 0_R$ for all $r \in R$.
\end{itemize}
We say that a semiring $R$ is \emph{idempotent} when $r + r = r$ for every $r \in R$, \emph{i.e.}, when the additive monoid $(R,+)$ is idempotent (on the other hand, there is no additional constraint on the multiplicative monoid $(R,\cdot)$).

\begin{example}\label{ex:bgen:semiring}
  A key example of an infinite idempotent semiring is the set $2^{A^*}$ of all languages over~$A$. Union is the addition (with $\emptyset$ as neutral element) and language concatenation is the multiplication (with $\{\varepsilon\}$ as neutral element).
\end{example}

Clearly, any finite idempotent semiring $(R,+,\cdot)$ is in particular a rating set: $(R,+)$ is an idempotent and commutative monoid. In particular, one may verify that the canonical ordering ``$\leq$'' on $R$, is compatible with multiplication as well.

\medskip

We may now define \mratms: as expected they are semiring morphisms. Let $\rho: 2^{A^*} \to R$ be a \ratm. By definition, this means that the rating set $(R,+)$ is a finite idempotent commutative monoid and that $\rho$ is a monoid morphism from $(2^{A^*},\cup)$ to $(R,+)$. We say that $\rho$ is \tame when the rating set $R$ is equipped with a second binary operation ``$\cdot$'' such that $(R,+,\cdot)$ is an idempotent semiring and $\rho$ is also a monoid morphism from $(2^{A^*},\cdot)$ to $(R,\cdot)$. In other words, the two following additional axioms hold:
\begin{enumerate}\setcounter{enumi}{2}
\item\label{itm:bgen:funit} $\rho(\varepsilon) = 1_R$.
\item\label{itm:bgen:fmult} For all $K_1,K_2 \subseteq A^*$, we have $\rho(K_1K_2) = \rho(K_1) \cdot \rho(K_2)$.
\end{enumerate}

\begin{remark}
  An important point is that the \ratms which are both \nice and \tame are finitely representable. As we explained above, a \nice \ratm $\rho: 2^{A^*} \to R$ is characterized by the canonical map $\rho_*: A^* \to R$.  Moreover, when $\rho$ is \tame as well, $\rho_*$ is finitely representable: it is a morphism from $A^{*}$ into the finite monoid $(R,{\cdot})$. Thus, we may consider algorithms taking \nice \mratms as input. Let us point out that the \ratms which are not \nice and \tame remain important. We often deal with them in our proofs.
\end{remark}

\subsection{\Imprints and optimal covers}

We now explain how we use \ratms to measure the quality of covers. This involves an additional notion: ``\imprints''. Consider a \ratm $\rho: 2^{A^*} \to R$. For any finite set of languages \Kb, the \emph{$\rho$-\imprint} of \Kb, denoted by $\prin{\rho}{\Kb}$, is the following subset of~$R$:
\[
  \begin{array}{lll}
    \prin{\rho}{\Kb} & = & \dclosr \{\rho(K) \mid K\in\Kb\} \\
                  & = & \{r \in R \mid \text{there exists $K \in \Kb$ such that $r \leq \rho(K)$}\}.
  \end{array}
\]
When using this notion, we shall have some language $L \subseteq A^*$ in hand: our goal is to find the ``best possible'' cover \Kb of $L$. Intuitively, $\rho$-\imprints measure the ``quality'' of candidate covers \Kb (the smaller the $\rho$-\imprint, the better the quality).

\medskip

This leads to the notion of optimality. Let \Ds be an arbitrary lattice. Given a language~$L$, an \emph{optimal \Ds-cover of~$L$ for $\rho$} is a \Ds-cover of $L$ which has the smallest possible $\rho$-\imprint (with respect to inclusion). That is, \Kb is an optimal \Ds-cover of $L$ for $\rho$ if and only if,
\[
  \prin{\rho}{\Kb} \subseteq \prin{\rho}{\Kb'} \quad \text{for every \Ds-cover $\Kb'$ of $L$}.
\]
Furthermore, in the special case when $L = A^*$, we speak of optimal \emph{universal} \Ds-cover for $\rho$.

\medskip

If \Ds is a lattice, one can show that there always exists at least one optimal \Ds-cover of~$L$ for $\rho$ (see~\cite[Lemma~4.15]{pzcovering2}). In general, there are actually infinitely many of~them.

\begin{lemma}\label{lem:bgen:opt}
  Let \Ds be a lattice. Then, for any \ratm $\rho: 2^{A^*} \to R$ and any language $L \subseteq A^*$, there exists an optimal \Ds-cover of $L$ for $\rho$.
\end{lemma}

It is important to note that the proof of Lemma~\ref{lem:bgen:opt} is non-constructive: given $L$ and $\rho: 2^{A^*} \to R$, computing an actual optimal \Ds-cover of $L$ for $\rho$ is a difficult problem in general. As we shall see below, getting such an algorithm (in the special case when $\rho$ is \nice and \tame) yields a procedure for \Ds-covering. Before we can establish this connection precisely, we require a key observation about optimal \Ds-covers.

\medskip
\noindent{\bf Optimal \imprints.} By definition, given a language $L$, all optimal \Ds-covers of $L$ for $\rho$ have the same $\rho$-\imprint. Hence, this unique $\rho$-\imprint is a \emph{canonical} object for \Ds, $L$ and $\rho$. We call it the \emph{\Ds-optimal $\rho$-\imprint on $L$} and we denote it by $\opti{\Ds}{L,\rho}$:
\[
  \opti{\Ds}{L,\rho} = \prin{\rho}{\Kb} \subseteq R \quad \text{for every optimal \Ds-cover \Kb of $L$ for $\rho$}.
\]
Additionally, in the particular case when $L = A^*$, we shall speak of \emph{\Ds-optimal universal $\rho$-\imprint} and write \opti{\Ds}{\rho} for \opti{\Ds}{A^*,\rho}.

\medskip

We complete these definitions with a few properties of optimal \imprints. We start with a straightforward fact which compares the optimal \imprints on languages which are comparable with inclusion. The proof is available in~\cite[Fact~4.17]{pzcovering2}.

\begin{fct}\label{fct:inclus2}
  Consider a \ratm $\rho: 2^{A^*} \to R$ and a lattice \Ds. Let $H,L$ be two languages such that $H \subseteq L$. Then, $\opti{\Ds}{H,\rho} \subseteq \opti{\Ds}{L,\rho}$.
\end{fct}

\noindent
More precisely, the following fact connects optimal \imprints with union of languages.

\begin{fct}\label{fct:lattice:optunion}
  Let $\rho: 2^{A^*} \to R$ be a \ratm and consider two languages $H,L$. Then, for every lattice \Ds, we have $\opti{\Ds}{H \cup L,\rho} = \opti{\Ds}{H,\rho} \cup \opti{\Ds}{L,\rho}$.
\end{fct}

\begin{proof}
  We already know that $\opti{\Ds}{H,\rho} \subseteq \opti{\Ds}{H \cup L,\rho}$ and $\opti{\Ds}{L,\rho} \subseteq \opti{\Ds}{H \cup L,\rho}$ by Fact~\ref{fct:inclus2}. Therefore, the inclusion $\opti{\Ds}{H,\rho} \cup \opti{\Ds}{L,\rho} \subseteq \opti{\Ds}{H \cup L,\rho}$ is immediate. We prove the converse one. Let $r \in \opti{\Ds}{H \cup L,\rho}$. We let $\Kb_H$ and $\Kb_L$ be optimal \Ds-covers of $H$ and $L$ respectively (for $\rho$). Clearly, $\Kb_H \cup \Kb_L$ is a cover $H \cup L$. Therefore, $r \in \opti{\Ds}{H \cup L,\rho}$ implies that $r \in \prin{\rho}{\Kb_H\cup \Kb_L}$. Hence, there exists $K \in \Kb_H\cup \Kb_L$ such that $r \leq \rho(K)$. Then, either $K \in \Kb_H$ which implies $r \in \prin{\rho}{\Kb_H} = \opti{\Ds}{H,\rho}$ or $K \in \Kb_L$ which implies $r \in \prin{\rho}{\Kb_L} = \opti{\Ds}{L,\rho}$. Altogether, we get $r \in \opti{\Ds}{H,\rho} \cup \opti{\Ds}{L,\rho}$, finishing the proof.
\end{proof}

\subsection{Connection with the covering problem} Finally, we explain how \ratms are used for handling the covering problem.

Given a lattice \Ds, it turns out that the \Ds-covering problem reduces to another problem whose input is a \nice \mratm. Let us point out that \emph{two} reductions of this kind are presented in~\cite{pzcovering2}. The first one is simpler but restricted to Boolean algebras. On the other hand, the second one applies to any lattice, but requires working with more involved objects.

In the paper, we investigate classes of the form \bpol{\Ds}, which are Boolean algebras. Hence, we shall mostly work with the first variant whose statement is as follows.

\begin{propC}[\cite{pzcovering2}]\label{prop:thereduction}
  Let \Cs be a Boolean algebra. Assume that there exists an algorithm for the following computational problem:
  \begin{center}
    \begin{tabular}{ll}
      {\bf Input:}  & A \nice \mratm $\rho:2^{A^*} \to R$.                  \\
      {\bf Output:} & Compute the \Cs-optimal universal $\rho$-\imprint, \opti{\Cs}{\rho}.
    \end{tabular}
  \end{center}
  Then, \Cs-covering is decidable.
\end{propC}

\begin{proof}[Proof sketch]
  The proof builds on several simple steps. The first one is the observation that when $\Cs$ is a Boolean algebra, deciding whether a pair $(L,\Lb)$ is \Cs-coverable reduces to the case where $L=A^{*}$. Indeed, one may check that $(L,\Lb)$ is \Cs-coverable if and only if so is $(A^*, \Lb\cup\{L\})$. We are left to show that one can decide whether a pair $(A^{*},\Lb)$ is~\Cs-coverable.

  Second, observe that $(2^{\Lb},\cup,\emptyset)$ is a rating set. One may verify that the function $\rho_{\Lb}:2^{A^*}\to2^{\Lb}$ defined by $\rho_{\Lb}(K)=\{L\in\Lb\mid K\cap L\neq\emptyset\}$ is a \nice \ratm. Furthermore, for any subset \Hb of \Lb, one may prove that $(A^{*},\Hb)$ is \Cs-coverable if and only if $\Hb\in \opti{\Cs}{\rho_\Lb}$.

  The third step is to reduce the computation of $\opti{\Cs}{\rho_\Lb}$ to that of $\opti{\Cs}{\rho'_\Lb}$, where $\rho'_{\Lb}$ is a \emph{\tame} \nice \ratm that we can build from~$\rho_{\Lb}$. By the hypothesis of the statement, this last set is computable. We refer the reader to~\cite[Theorem 5.21]{pzcovering2} for more details.
\end{proof}

Additionally, we shall need in Section~\ref{sec:tools} to apply a theorem of~\cite{pseps3j} for classes of the form \pol{\Cs} (which are lattices but not Boolean algebras) as a sub-result. Thus, we also recall the terminology associated to the generalized reduction which holds for arbitrary lattices, since we need it to state this theorem.

When working with an arbitrary lattice, one needs to consider slightly more involved objects. Given a lattice \Ds, a map $\alpha: A^* \to M$ into a finite set $M$ (in practice, $\alpha$ will be a monoid morphism but this is not required for the definition) and a \ratm $\rho:2^{A^*} \to R$, we write \popti{\Ds}{\alpha}{\rho} for the following set,
\[
  \popti{\Ds}{\alpha}{\rho} = \{(s,r) \in M \times R \mid r \in \opti{\Ds}{\alpha\inv(s),\rho}\} \subseteq M \times R.
\]
We call \popti{\Ds}{\alpha}{\rho} the \emph{$\alpha$-pointed \Ds-optimal $\rho$-\imprint}. Clearly, it encodes all sets $\opti{\Ds}{\alpha\inv(s),\rho}$ for $s \in M$.

\smallskip
The following statement is~\cite[Proposition~5.18]{pseps3j} (see also~\cite[Proposition~7.2]{pzcovering2}).

\begin{proposition}\label{prop:thereductionlat}
  Consider a lattice \Ds and some finite \vari \Cs. Assume that there exists an algorithm for the following computational problem:
  \begin{center}
    \begin{tabular}{ll}
      {\bf Input:}  & A \Cs-compatible morphism $\alpha: A^*\to M$ and                                    \\
                    & a \nice \mratm $\rho:2^{A^*} \to R$.                                                \\
      {\bf Output:} & Compute the $\alpha$-pointed \Ds-optimal $\rho$-\imprint, \popti{\Ds}{\alpha}{\rho}.
    \end{tabular}
  \end{center}
  Then, \Ds-covering is decidable.
\end{proposition}

\section{Characterization of \bpol{\Cs}-optimal \imprints}
\label{sec:bpol}
We present a generic characterization of \bpol{\Cs}-optimal \imprints which holds when \Cs is a finite \vari. For the sake of avoiding clutter, we assume that the finite \vari \Cs is fixed for the whole section.

Given a \emph{\nice} \mratm $\rho: 2^{A^*} \to R$, we want to characterize the set $\bpopti \subseteq R$. An important point is that we do not work directly with this set. Instead, we characterize the family of all sets $\opti{\bpol{\Cs}}{D,\rho} \subseteq R$ where $D \subseteq A^*$ is a $\sim_{\Cs}$-class. Note that this family of sets record more information than just the set \bpopti. Indeed, by Fact~\ref{fct:lattice:optunion}, we have,
\[
  \bpopti = \opti{\bpol{\Cs}}{A^*,\rho} = \bigcup_{D \in \sclac} \opti{\bpol{\Cs}}{D,\rho}.
\]
For the sake of convenience, we shall encode this family as a set of pairs in $(\sclac) \times R$.  Given a \mratm $\rho: 2^{A^* } \to R$, we define:
\[
  \cbpopti = \{(D,r) \in (\sclac) \times R \mid r \in \opti{\bpol{\Cs}}{D,\rho}\}.
\]
When $\rho$ is \nice, we characterize \cbpopti as the \emph{greatest} subset of $R$ satisfying specific properties. From the statement, it is straightforward to obtain a greatest fixpoint procedure for computing \cbpopti from $\rho$. In turns, this allows to compute \bpopti using the above equality. By Proposition~\ref{prop:thereduction}, this yields an algorithm for \bpol{\Cs}-covering, thus proving our main result: Theorem~\ref{thm:bool}.

\begin{remark}
  This characterization is rather unique among the results that have been obtained for other classes in~\cite{pzcovering2} and~\cite{pseps3j}. Typically, optimal \imprints are characterized as \textbf{least} subsets, not greatest ones.
\end{remark}

\begin{notation}
  In our statements, we shall frequently manipulate subsets of $(\sclac) \times R$. When doing so, the following notation will be convenient. Given $S \subseteq (\sclac) \times R$ and $D \in \sclac$, we write,
  \[
    S(D) = \{r \in R \mid (D,r) \in S\}
  \]
  In particular, observe that $\cbpopti(D) = \opti{\bpol{\Cs}}{D,\rho}$ by definition.
\end{notation}

Given a \mratm $\rho: 2^{A^*} \to R$, we define a notion of \bpol{\Cs}-saturated subset of $(\sclac) \times R$ (for $\rho$). Our theorem then states that when $\rho$ is \nice, the greatest such subset is exactly \cbpopti.

\begin{remark}
  The definition of \bpol{\Cs}-saturated sets makes sense regardless of whether $\rho$ is \nice. However, we need this hypothesis for the greatest one to be \cbpopti.
\end{remark}

The definition is based on an intermediary notion. With every set $S \subseteq (\sclac) \times R$, we associate another set $\rbpols \subseteq (\sclac) \times R \times 2^R$. For the definition, we need to recall a few properties. Given a word $w \in A^*$, we denote its $\sim_{\Cs}$ class by \ctype{w}. Moreover, since \Cs is closed under quotients, Lemma~\ref{lem:canoquo} yields that the equivalence $\sim_{\Cs}$ is a congruence. We denote by ``$\cmult$'' the multiplication of $\sim_\Cs$-classes in the monoid \sclac. Additionally, since $R$ is a semiring, $2^R$ is one as well for union as addition and the natural multiplication lifted from the one of $R$ (for $U,V \in 2^R$,  $UV = \{qr \mid q \in U \text{ and } r \in V\}$).

We may now present our definition. Consider a set $S \subseteq (\sclac) \times R$. We define \rbpols as the least subset of $(\sclac) \times R \times 2^R$ (with respect to inclusion) which satisfies the following properties:
\begin{itemize}
\item {\it Trivial elements.} For every $w \in A^*$, we have $(\ctype{w},\rho(w),\{\rho(w)\}) \in \rbpols$.
\item{\it Extended downset.} For every $(C,q,U) \in \rbpols$ and $V \subseteq \dclosr U$, we have $(C,q,V) \in \rbpols$.
\item {\it Multiplication.} For every $(C,q,U),(D,r,V) \in \rbpols$, we have $(C \cmult D,qr,UV) \in \rbpols$.
\item {\it $S$-restricted closure.} For every triple of idempotents $(E,f,F) \in \rbpols$, we have\\ $(E,f,F \cdot S(E) \cdot F) \in \rbpols$.
\end{itemize}

We are ready to define the \bpol{\Cs}-saturated subsets of $(\sclac) \times R$. Consider a set $S \subseteq (\sclac) \times R$. We say that $S$ is \emph{\bpol{\Cs}-saturated for $\rho$} if the following property holds:
\begin{equation}\label{eq:bpol:sat}
\begin{array}{c}
\text{for every $(D,r) \in S$, there exist $r_1,\dots,r_k \in R$ such that,} \\
\text{$r \leq r_1 + \cdots + r_k$ and $(D,r_i,\{r_1 + \cdots + r_k\}) \in \rbpols$ for every $i \leq k$}
\end{array}
\end{equation}

We now state our characterization of \bpol{\Cs}-optimal \imprints in the following theorem.

\begin{theorem}\label{thm:bpol:carac}
  Let $\rho: 2^{A^*} \to R$ be a \nice \mratm. Then, \cbpopti is the greatest \bpol{\Cs}-saturated subset of $(\sclac) \times R$ for $\rho$.
\end{theorem}

Given as input a \nice \mratm $\rho:2^{A^*} \to R$, Theorem~\ref{thm:bpol:carac} yields an algorithm to compute \bpopti. Indeed, computing the greatest \bpol{\Cs}-saturated subset of $(\sclac) \times R$ is achieved with a greatest fixpoint algorithm. One starts from the set $S_0 = (\sclac) \times R$ and computes a sequence $S_0 \supseteq S_1 \supseteq S_2 \supseteq \cdots$ of subsets. For every $n \in \nat$, $S_{n+1}$ is the set of all pairs $(D,r) \in S_n$ satisfying~\eqref{eq:bpol:sat} for $S = S_n$. That is,
\[
  \begin{array}{c}
    \text{there exist $r_1,\dots,r_k \in R$ such that,} \\
    \text{$r \leq r_1 + \cdots + r_k$ and $(D,r_i,\{r_1 + \cdots + r_k\}) \in \rbpol{S_{n}}$ for every $i \leq k$}
  \end{array}
\]
Clearly, $S_{n+1} \subseteq S_n$ for every $n \in \nat$. Therefore, the computation eventually reaches a fixpoint which is the greatest \bpol{\Cs}-saturated subset of $(\sclac) \times R$ by definition. Let us point out that the computation of $S_{n+1}$ from $S_n$ involves computing \rbpol{S_{n}}, which is achieved with a least fixpoint procedure by definition.

Altogether, it follows that Theorem~\ref{thm:bpol:carac} yields an algorithm for computing \cbpopti (and therefore, \bpopti as well by Fact~\ref{fct:lattice:optunion}) which alternates between a greatest fixpoint and a least fixpoint.

\begin{remark}\label{rem:niceiscrucial}
  The hypothesis that $\rho$ is \nice in Theorem~\ref{thm:bpol:carac} is mandatory: the result fails otherwise. This is actually apparent on the definition of \bpol{\Cs}-saturated sets. One may verify from the definition of \rbpols that for every triple $(D,q,U) \in \rbpols$, there exists a word $w \in A^*$ such that $D = \ctype{w}$ and $q = \rho(w)$. Therefore, it follows from~\eqref{eq:bpol:sat} that for every $(D,r) \in (\sclac)\times R$ belonging to a \bpol{\Cs}-saturated subset, its second component $r$ must satisfy $r \leq \rho(w_1) + \cdots + \rho(w_k)$ for some words $w_1,\dots,w_k \in A^*$. Intuitively, this means that \bpol{\Cs}-saturated subsets only depend on the image of singletons. Therefore, using the notion only makes sense when $\rho$ is characterized by these images: this is exactly the definition of \nice \ratms.

  This might seem to be a minor observation. Indeed, by Proposition~\ref{prop:thereduction}, being able to compute \bpopti from a \nice \mratm suffices to meet our goal: getting an algorithm for \bpol{\Cs}-covering. Actually, it does not even make sense to speak of an algorithm which takes arbitrary \mratms as input since we are not able to finitely represent them. However, from a theoretical point of view, the fact that we only manage to get a description of \bpopti when $\rho$ is \nice is significant. In the proof of Theorem~\ref{thm:bpol:carac}, we use a theorem of~\cite{pseps3j} as a subresult. Specifically, this theorem is a characterization of \pol{\Cs}-optimal pointed \imprints: given a \Cs-compatible morphism $\alpha: A^* \to M$ and a \mratm $\tau: 2^{A^*} \to Q$, it describes the set \popti{\pol{\Cs}}{\alpha}{\tau}. A key point is that this characterization does \textbf{not} require $\tau$ to be \nice. This is crucial: in the proof of Theorem~\ref{thm:bpol:carac}, we consider auxiliary \ratms built from $\rho$ which need \textbf{not} be~\nice.

  In summary, we are able to handle \pol{\Cs} for all \mratms, including those that are not \nice, which is crucial in order be able to handle \bpol{\Cs}. However, at this level, we are only able to deal with \nice \mratms (the situation is actually similar for \pol{\bpol{\Cs}} as shown in~\cite{pseps3j}). This explains why the results presented in this paper and in~\cite{pseps3j} cannot be lifted to higher levels in concatenation hierarchies (at least not in a straightforward manner).
\end{remark}

We turn to the proof of Theorem~\ref{thm:bpol:carac}. It spans the remaining four sections of the paper. Given a \nice \mratm $\rho: 2^{A^*} \to R$, we have to show that \cbpopti is the greatest \bpol{\Cs}-saturated subset of $(\sclac) \times R$ for $\rho$. The main argument involves two directions which are proved independently. They correspond to soundness and completeness of the greatest fixpoint algorithm computing \cbpopti.
\begin{itemize}
\item The soundness argument shows that \cbpopti contains every \bpol{\Cs}-saturated subset (this implies that the greatest fixpoint procedure only computes elements of \cbpopti). We present it in Section~\ref{sec:sound}.
\item The completeness argument shows that \cbpopti itself is \bpol{\Cs}-saturated (this implies that the greatest fixpoint procedure computes all elements of \cbpopti). We present it in Section~\ref{sec:comp}.
\end{itemize}
When put together, these two results yield as desired that \cbpopti is the greatest \bpol{\Cs}-saturated subset of $R$, proving Theorem~\ref{thm:bpol:carac}.

However, before presenting the main argument, we require some additional material about \ratms. For both directions, we shall introduce auxiliary \ratms (built from~$\rho$) and apply a characterization of \pol{\Cs}-optimal \imprints to them (taken from~\cite{pseps3j}). These auxiliary \ratms are built using generic constructions which are not specific to Theorem~\ref{thm:bpol:carac}. We present them in Section~\ref{sec:tools}. Then, we recall the theorem characterizing  \pol{\Cs}-optimal \imprints from~\cite{pseps3j} in Section~\ref{sec:pol} (actually, we slightly generalize it, since we shall apply it for \ratms that are more general than the ones considered in~\cite{pseps3j}).

\section{Nesting of \ratms}
\label{sec:tools}
In this section, we present two generic constructions. Each of them builds a new \ratm out of an already existing one and a lattice \Ds. The constructions are new: they do not appear in~\cite{pzcovering2} (however, one of them generalizes and streamlines a technical construction used in~\cite{pseps3j}).

\begin{remark}
  As announced, we shall later rely on these constructions in the proof of Theorem~\ref{thm:bpol:carac} (we use them in the special cases when \Ds is either \pol{\Cs} or \bpol{\Cs}). However, this section is independent from Theorem~\ref{thm:bpol:carac}: all definitions are presented in a general context.
\end{remark}

We first present the constructions and then investigate the properties of the output \ratms they produce.

\subsection{Definition}

We present two constructions. The first one involves two objects: a lattice \Ds and a \ratm $\rho: 2 ^{A^*} \to R$. We build a new \ratm \bratauxd whose rating set is $(2^R,\cup)$, and which associates to a language its optimal \Ds-optimal $\rho$-imprint. \[
  \begin{array}{llll}
    \bratauxd: & (2^{A^*},\cup) & \to     & (2^R,\cup)         \\
               & K              & \mapsto & \opti{\Ds}{K,\rho}
  \end{array}
\]

The fact that \bratauxd is indeed a \ratm is shown below in Proposition~\ref{prop:areratms}.
The second construction involves an additional object: a map $\alpha: A^* \to M$ where $M$ is some arbitrary finite set (in practice, $\alpha$ will be a monoid morphism, but this is not required for the definition). We build a new \ratm \lratauxd with rating set $(2^{M \times R},\cup)$:\label{rataux}
\[
  \begin{array}{llll}
    \lratauxd: & (2^{A^*},\cup) & \to     & (2^{M \times R},\cup)                                      \\
               & K              & \mapsto & \{(s,r) \mid r \in \opti{\Ds}{\alpha\inv(s) \cap K,\rho}\}.
  \end{array}
\]
Let us prove that these two maps are indeed \ratms. We state this result in the following proposition.

\begin{proposition} \label{prop:areratms}
  Consider a lattice \Ds, a map $\alpha: A^* \to M$ into a finite set $M$ and a \ratm $\rho: 2 ^{A^*} \to R$. Then, \bratauxd and \lratauxd are \ratms.
\end{proposition}

\begin{proof}
  We start with $\bratauxd$. It is immediate that $\bratauxd(\emptyset) =  \opti{\Ds}{\emptyset,\rho} = \emptyset$. Moreover, we obtain from Fact~\ref{fct:lattice:optunion} that for every $H,L \subseteq A^*$,
  \[
    \bratauxd(H \cup L) = \opti{\Ds}{H \cup L,\rho} = \opti{\Ds}{H,\rho} \cup \opti{\Ds}{L,\rho} = \bratauxd(H) \cup \bratauxd(L).
  \]
  We conclude that \bratauxd is indeed a \ratm. We turn to \lratauxd. Clearly,
  \[
    \lratauxd(\emptyset) = \{(s,r) \mid r \in \opti{\Ds}{\emptyset,\rho}\} =  \{(s,r) \mid r \in \emptyset\} = \emptyset.
  \]
  Moreover, given $H,L \subseteq A^*$,
  \[
    \lratauxd(H \cup L) = \{(s,r) \mid r \in \opti{\Ds}{\alpha\inv(s) \cap (H \cup L),\rho}\}.
  \]
  By Fact~\ref{fct:lattice:optunion}, this yields,
  \[
    \lratauxd(H \cup L) = \{(s,r) \mid r \in \opti{\Ds}{\alpha\inv(s) \cap H,\rho} \cup \opti{\Ds}{\alpha\inv(s) \cap L,\rho}\}
  \]
  This exactly says that $\lratauxd(H \cup L) = \lratauxd(H) \cup \lratauxd(L)$, finishing the proof that \lratauxd is a \ratm.
\end{proof}

A crucial observation is that the \ratms \bratauxd and \lratauxd are not \nice in general, even when the original \ratm $\rho$ is. Let us present a counter-example.

\begin{example}\label{ex:notnice}
  Let \Ds be the Boolean algebra consisting of all languages which are either finite or co-finite (\emph{i.e.}, their complement is finite). Moreover, let $T = \{0,1\}$ and $R = 2^T$. We define a \nice \ratm $\rho: 2^{A^*} \to R$ as follows (actually, it is simple to verify from the definition that $\rho$ is also \tame). Since we are defining a \nice \ratm, it suffices to specify the evaluation of words: for any $w \in A^*$, we let $\rho(w) = \{0\}$ if $w$ has even length and $\rho(w) = \{1\}$ if $w$ has odd length. We show that the \ratm $\bratauxd: 2^{A^*} \to 2^R$ is not \nice.

  By definition, $\bratauxd(A^*) = \opti{\Ds}{A^*,\rho}$. Recall that \Ds contains only finite and co-finite languages. Moreover, covers may only contain finitely many languages. Hence, it is immediate that if \Kb is an optimal \Ds-cover of $A^*$ for $\rho$, then there exists $K \in \Kb$ containing a word of even length and a word of odd length. Therefore $\rho(K) = \{0,1\}$ by definition of $\rho$, whence $\bratauxd(A^*) = \prin{\rho}{\Kb} = \{\{0,1\},\{0\},\{1\},\emptyset\} = R$.

  Now observe that for any $w \in A^*$, $\{w\} \in \Ds$ by definition (it is a finite language). Hence, $\{\{w\}\}$ is a \Ds-cover of $\{w\}$ and since $\bratauxd(w) = \opti{\Ds}{\{w\},\rho}$, we know that $\bratauxd(w) = \{\{0\},\emptyset\}$ if $w$ has even length and $\bratauxd(w) = \{\{1\},\emptyset\}$ if $w$ has odd length. Altogether, we obtain that,
  \[
    \bigcup_{w \in A^*} \bratauxd(w) = \{\{0\},\{1\},\emptyset\} \neq \bratauxd(A^*)
  \]
  We conclude that \bratauxd is not \nice.
\end{example}

Another important question is whether \bratauxd and \lratauxd are \tame. The remainder of the section is devoted to discussing this point.

\subsection{Multiplication}

If $\alpha: A^* \to M$ is a morphism into a finite monoid and $\rho: 2^{A^*} \to R$ is a \mratm, $(M,\cdot)$ is a monoid and $(R,+,\cdot)$ is an idempotent semiring. We may lift the multiplication of $R$ to $2^R$ in the natural way: given $U,V \in 2^R$, we let $UV = \{qr \mid q \in U \text{ and } r \in V\}$. One may verify that $(2^R,\cup,\cdot)$ is an idempotent semiring. Similarly, we may lift the componentwise multiplication on $M \times R$ to $2^{M \times R}$ and $(2^{M \times R},\cup,\cdot)$ is an idempotent semiring. Whenever we consider semiring structures for $2^R$ and $2^{M \times R}$, these are the additions and multiplications that we shall use.

Unfortunately, even though $2^R$ and $2^{M\times R}$ are semirings, \textbf{neither} $\bratauxd: 2^{A^*} \to 2^R$ nor $\lratauxd: 2^{A^*} \to 2^{M \times R}$ are \tame: they are not monoid morphisms for multiplication. However, it turns out that they behave almost as \mratms when the class \Ds satisfies appropriate properties related to closure under concatenation. We formalize this with a new notion: quasi-\mratms.

\medskip
\noindent
{\bf Quasi-\mratms.} Let $\rho: (2^{A^*},\cup) \to (R,+)$ be a \ratm whose rating set is equipped with a multiplication ``$\cdot$'' such that $(R,+,\cdot)$ is a semiring (however, $\rho$ is not required to be \tame for this multiplication). Finally, let \quasir be an endomorphism of the additive monoid $(R,+)$ (\emph{i.e.}, $\quasir(0_R)=0_R$ and for all $r,s\in R$, $\quasir(r+s)=\quasir(r)+\quasir(s)$). We say that $\rho$ is \emph{quasi-\tame for \quasir} when the following axioms are satisfied:
\begin{enumerate}
	\item\label{itm:nest:qt1} For every $q,r,s \in R$, $\quasir(q\quasir(r)s) = \quasir(qrs)$.
	\item\label{itm:nest:qt3} For every $K_1,K_2 \subseteq A^*$, we have $\rho(K_1K_2) = \quasir(\rho(K_1) \cdot \rho(K_2))$.
\end{enumerate}
For the sake of improved readability, we often abuse terminology and simply say that ``the \ratm $\rho$ is quasi-\tame'', assuming implicitly that the endomorphism \quasir is defined and fixed. Note however that this is slightly ambiguous as there might be several endomorphisms of $(R,+)$ satisfying the above axioms for the same \ratm $\rho$.

\begin{remark}
	Since $\quasir$ is an endomorphism of $(R,+)$, it preserves the canonical order on $R$. Given $q,r \in R$ such that $q \leq r$, we have $\quasir(q) \leq \quasir(r)$. Indeed, by definition, $q \leq r$ means that $q+ r=r$. Therefore, since $\quasir$ is a morphism, $\quasir(q) + \quasir(r) = \quasir(r)$ which means that $\quasir(q) \leq \quasir(r)$. We use this property implicitly in proofs.
\end{remark}

\begin{remark} \label{rem:tameisquasi}
	Clearly, a true \mratm is always quasi-\tame. Indeed, in this case, it suffices to choose $\quasir$ as the identity; $\quasir(r) = r$ for all $r \in R$.
\end{remark}

We have the following useful fact about quasi-\mratms.

\begin{fct} \label{fct:nest:quasip}
	Let $\rho: 2^{A^*} \to R$ be a quasi-\mratm. For every language $H \subseteq A^*$, we have $\rho(H) = \quasir(\rho(H))$.
\end{fct}

\begin{proof}
	Axiom~\ref{itm:nest:qt3} in the definition yields $\rho(H) = \quasir(\rho(H) \cdot \rho(\veps))$. Thus, Axiom~\ref{itm:nest:qt1} yields $\quasir(\rho(H)) = \quasir(\quasir(\rho(H) \cdot \rho(\veps))) = \quasir(\rho(H) \cdot \rho(\veps)) = \rho(H)$. This concludes the proof.
\end{proof}

Additionally, we shall need the following lemma which is a straightforward adaptation of a result proved in~\cite[Lemma~5.8]{pzcovering2} to quasi-\mratms.

\begin{lemma} \label{lem:closmult}
  Let \Ds be a \pvari and $\rho: 2^{A^*} \to R$ be a quasi-\mratm. For every $H,L \subseteq A^*$, $q \in \opti{\Ds}{H,\rho}$ and $r \in \opti{\Ds}{L,\rho}$, we have $\quasir(qr) \in \opti{\Ds}{HL,\rho}$.
\end{lemma}

\begin{proof}
	Let $q \in \opti{\Ds}{H,\rho}$ and $r \in \opti{\Ds}{L,\rho}$. By definition, it suffices to prove that for every \Ds-cover \Kb of $HL$, we have $\quasir(qr) \in \prin{\rho}{\Kb}$. Let \Kb be a $\Ds$-cover of $HL$, we have to find $K \in \Kb$ such that $\quasir(qr) \leq \rho(K)$. We use the following claim which is based on the Myhill-Nerode theorem.

  \begin{claim}
    There exists a language $G \in \Ds$ which satisfies the following two properties:
    \begin{enumerate}
    \item For all $u \in H$, there exists $K \in \Kb$ such that $G \subseteq u\inv K$.
    \item $r \leq \rho(G)$
    \end{enumerate}
  \end{claim}

  \begin{proof}[Proof of the claim]
    For every $u \in H$, we let $\Qb_u = \{u\inv K \mid K \in \Kb\}$. Clearly, $\Qb_u$ is a \Ds-cover of $L$ since $\Kb$ is a cover of $HL$ and \Ds is closed under quotients. Moreover, we know by hypothesis on \Ds that all languages in  \Kb are regular. Therefore, it follows from the Myhill-Nerode theorem that they have finitely many quotients. Thus, while there might be infinitely many words $u \in H$, there are only finitely many distinct sets $\Qb_u$. It follows that we may use finitely many intersections to build a \Ds-cover \Qb of $L$ such that for every $Q \in \Qb$ and every $u \in H$, there exists $K \in \Kb$ satisfying $Q \subseteq u\inv K$. This means that all $Q \in \Qb$ satisfy the first item in the claim, we now pick one which satisfies the second one as well.

    Since $r \in \opti{\Ds}{L,\rho}$, and \Qb is a \Ds-cover of $L$, we have $r \in \prin{\rho}{\Qb}$. Thus, we get $G \in \Qb$ such that $r \leq \rho(G)$ by definition. This concludes the proof of the claim.
  \end{proof}

  We may now finish the proof of Lemma~\ref{lem:closmult}. Let $G \in \Ds$ be as defined in the claim and consider the following set:
  \[
    \Gb = \left\{\bigcap_{v\in G}Kv\inv \mid K \in \Kb\right\}.
  \]
  Observe that all languages in \Gb belong to \Ds. Indeed, by hypothesis on \Ds, every $K \in \Kb$ is regular. Thus, it has finitely many right quotients by the Myhill-Nerode theorem and the language $\bigcap_{v\in H}Kv\inv$ is the intersection of finitely many quotients of languages in \Ds. By closure under intersection and quotients, it follows that $\bigcap_{v\in G}Kv\inv \in \Ds$. Moreover, \Gb is a \Ds-cover of $H$. Indeed, given $u \in H$, we have $K \in \Kb$ such that $G \subseteq u\inv K$ by the first assertion in the claim. Hence, for every $v \in G$, we have $u \in Kv\inv$ and we obtain that $u \in \bigcap_{v\in G}Kv\inv$, which is an element of \Gb.

  Therefore, since $q \in \opti{\Ds}{H,\rho}$ by hypothesis, we have $q \in \prin{\rho}{\Gb}$ and we obtain $G' \in \Gb$ such that $q \leq \rho(G')$. Hence, since $r \leq \rho(G)$ by the second item in the claim, we have $qr \leq \rho(G') \cdot \rho(G)$. Since $\rho$ is quasi-\tame over \Ds and $G,G' \in \Ds$, it follows from Axiom~\ref{itm:nest:qt3} in the definition of quasi-\mratms that,
  \[
    \rho(G'G) = \quasir(\rho(G') \cdot \rho(G))
  \]
  Moreover, since \quasir is an endomorphism of $(R,+)$ and $qr \leq \rho(G') \cdot \rho(G)$, we have by Fact~\ref{fct:increas}:
  \[
    \quasir(qr) \leq \quasir(\rho(G') \cdot \rho(G))
  \]
  Altogether, we get $\quasir(qr) \leq \rho(G'G)$. Finally, observe that $G'G \subseteq K$ for some $K \in \Kb$. Indeed, if $w \in G'G$, we have $w = uv$ with $u \in G'$ and $v \in G$. Moreover, $G' = \bigcap_{v\in G}Kv\inv$ for some $K \in \Kb$ by definition of \Gb. Hence, $u \in Kv\inv$ which yields $w = uv \in K$. Altogether, we get that $\quasir(qr) \leq \rho(G'G) \leq \rho(K)$, which concludes the proof.
\end{proof}

We now prove that when \Ds is a \pvari closed under concatenation, the \ratm \lratauxd is quasi-\tame provided that $\alpha$ is a morphism and $\rho$ is already quasi-\tame (actually, this is also true for \bratauxd but we do not need this result). This result is tailored to the situation in which we shall later use \lratauxd: $\Ds = \pol{\Cs}$.

\begin{lemma} \label{lem:copelrat}
  Let \Ds be a \pvari closed under concatenation, $\alpha: A^* \to M$ be a morphism and $\rho: 2 ^{A^*} \to R$ be a quasi-\mratm. Then, \lratauxd is quasi-\tame for the following associated endomorphism $\quasi{\lratauxd}$ of $(2^{M \times R},\cup)$:
  \[
    \quasi{\lratauxd}(T) = \dclosr \{(s,\quasir(r)) \mid (s,r) \in T\} \quad \text{for every $T \in 2^{M \times R}$}.
  \]
\end{lemma}

\begin{proof}
  We already know from Proposition~\ref{prop:areratms} that \lratauxd is a \ratm. Hence, we have to prove that the axioms of quasi-\mratms hold for the endomorphism $\quasi{\lratauxd}$ of $(2^{M \times R},\cup)$ described in the lemma (it is clear from the definition that this is indeed an endomorphism). For the sake of avoiding clutter, we write $\mu$ for \quasi{\lratauxd}.

  We start with the first axiom. Consider $T,U,V \in  2^{M \times R}$. We have to show that $\mu(T \mu(U) V) = \mu(TUV)$. Assume first that $(s,r) \in \mu(T \mu(U) V)$. By definition of $\mu$, this yields $(s_1,r_1) \in T$, $(s_2,r_2) \in \mu(U)$ and $(s_3,r_3) \in V$ such that $s= s_1s_2s_3$ and $r \leq \quasir(r_1r_2r_3)$. Since $(s_2,r_2) \in \mu(U)$, we have $(s_2,r'_2)  \in U$ such that $r_2 \leq \quasir(r'_2)$. It follows that $r_1r_2r_3 \leq r_1 \quasir(r'_2) r_2$ and since \quasir is an endomorphism of $(R,+)$, we obtain,
  \[
    r \leq \quasir(r_1r_2r_3) \leq \quasir(r_1 \quasir(r'_2) r_2).
  \]
  Since $\rho$ is quasi-\tame, the first axiom in the definition yields that $\quasir(r_1 \quasir(r'_2) r_2) = \quasir(r_1r'_2r_3)$ and we get $r \leq \quasir(r_1r'_2r_3)$. Since $(s_1,r_1) \in T$, $(s_2,r'_2)  \in U$ and $(s_3,r_3) \in V$. This yields $(s,r) = (s_1s_2s_3,r) \in \mu(TUV)$.

  Conversely, assume that $(s,r) \in \mu(TUV)$. By definition of $\mu$, we get $(s_1,r_1) \in T$, $(s_2,r_2) \in U$ and $(s_3,r_3) \in V$ such that $s= s_1s_2s_3$ and $r \leq \quasir(r_1r_2r_3)$. The first axiom in the definition of quasi-\mratms yields that $\quasir(r_1 \quasir(r_2) r_2) = \quasir(r_1r_2r_3)$. Therefore, $r \leq \quasir(r_1\quasir(r_2)r_3)$. Moreover, since $(s_2,r_2) \in U$, it is immediate that $(s_2,\quasir(r_2)) \in \mu(U)$ by definition of $\mu$. Altogether, this implies that $(s,r) \in \mu(T \mu(U) V)$.

  \smallskip

  It remains to establish that  \lratauxd fulfills the Axiom~\ref{itm:nest:qt3}. Let $K_1,K_2 \subseteq A^*$. We show that,
  \begin{equation}\label{eq:lratauxd:axiom3b}
    \lratauxd(K_1K_2)  = \mu(\lratauxd(K_1) \cdot \lratauxd(K_2)).
  \end{equation}
  We start with the right to left inclusion. Consider $(s,r) \in \mu\big(\lratauxd(K_1) \cdot \lratauxd(K_2)\big)$. By definition of $\mu$ we have $(s_1,r_1) \in \lratauxd(K_1)$ and $(s_2,r_2) \in \lratauxd(K_2)$ such that $s = s_1s_2$ and $r \leq \quasir(r_1r_2)$. By definition of $\lratauxd$, this means that $r_1 \in \opti{\Ds}{K_1 \cap \alpha\inv(s_1),\rho}$ and $r_2 \in \opti{\Ds}{K_2 \cap \alpha\inv(s_2),\rho}$. Since \Ds is a \pvari, Lemma~\ref{lem:closmult} yields that:
  \[
    \quasir(r_1r_2) \in \opti{\Ds}{(K_1 \cap \alpha\inv(s_1)) \cdot (K_2 \cap \alpha\inv(s_2)),\rho}.
  \]
  Observe that $(K_1 \cap \alpha\inv(s_1)) \cdot (K_2 \cap \alpha\inv(s_2)) \subseteq K_1K_2 \cap \alpha\inv(s_1)\alpha\inv(s_2)$. Moreover, since $\alpha$ is a morphism, it is clear that $\alpha\inv(s_1)\alpha\inv(s_2) \subseteq \alpha\inv(s_1s_2) =\alpha\inv(s)$. Altogether, this means that we have $(K_1 \cap \alpha\inv(s_1)) \cdot (K_2 \cap \alpha\inv(s_2)) \subseteq K_1K_2 \cap \alpha\inv(s)$.  Thus, Fact~\ref{fct:inclus2} yields $\quasir(r_1r_2) \in \opti{\Ds}{K_1K_2 \cap \alpha\inv(s),\rho}$. Since $r \leq \quasir(r_1r_2)$, this implies $r \in \opti{\Ds}{K_1K_2 \cap \alpha\inv(s),\rho}$, which exactly says that $(s,r) \in \lratauxd(K_1K_2)$, concluding the proof for the left to right inclusion.

  We turn to the converse inclusion. This is where we need \Ds to be closed under concatenation. Let $(s,r) \in \lratauxd(K_1K_2)$. We show that $(s,r) \in \mu\big(\lratauxd(K_1) \cdot \lratauxd(K_2)\big)$. By definition, we have $r \in \opti{\Ds}{K_1K_2 \cap \alpha\inv(s),\rho}$. For every $t \in M$ and $i \in \{1,2\}$, we define $\Hb_{i,t}$ as an optimal \Ds-cover of $K_i \cap \alpha\inv(t)$. Consider the following finite set of languages \Hb,
  \[
    \Hb = \{H_1H_2 \mid \text{there exist $s_1,s_2 \in M$ such that $s_1s_2 = s$, $H_1 \in \Hb_{1,s_1}$ and $H_2 \in \Hb_{2,s_2}$}\}.
  \]
  We prove that \Hb is a \Ds-cover of $K_1K_2 \cap \alpha\inv(s)$. Clearly all languages in \Hb belong to \Ds since \Ds is closed under concatenation by Theorem~\ref{thm:polclos}. Let us show that \Hb is a cover of $K_1K_2 \cap \alpha\inv(s)$. Consider $w \in K_1K_2 \cap \alpha\inv(s)$, we exhibit $H \in \Hb$ such that $w \in H$. Since $w \in K_1K_2$, we have $w = w_1w_2$ with $w_1 \in K_1$ and $w_2 \in K_2$. Let $s_1 = \alpha(w_1)$ and $s_2 = \alpha(w_2)$. Altogether, this means that $w_1 \in K_1 \cap \alpha\inv(s_1)$ and $w_2 \in K_2 \cap \alpha\inv(s_2)$. Therefore, we have $H_1 \in \Hb_{1,s_1}$ and $H_2 \in \Hb_{2,s_2}$ such that $w_1 \in H_1$ and $w_2 \in H_2$. This yields $w \in H_1H_2$. Finally, $s_1s_2 = \alpha(w) = s$ which yields that $H_1H_2 \in \Hb$ by definition.

  We may now finish the argument and show that $(s,r) \in \mu(\lratauxd(K_1) \cdot \lratauxd(K_2))$. Recall that $r \in \opti{\Ds}{K_1K_2 \cap \alpha\inv(s),\rho}$. Thus, since \Hb is a \Ds-cover of $K_1K_2 \cap \alpha\inv(s)$, we have $r \in \prin{\rho}{\Hb}$. It follows that there exists $H \in \Hb$ such that $r \leq \rho(H)$. By definition of \Hb, we have $H = H_1H_2$  with $H_1 \in \Hb_{1,s_1}$ and $H_2 \in \Hb_{2,s_2}$ where $s_1,s_2 \in M$ satisfy $s_1s_2 = s$. Let $r_1 = \rho(H_1)$ and $r_2 = \rho(H_2)$. Since $\Hb_{1,s_1}$ and $\Hb_{2,s_2}$ are optimal \Ds-covers of $K_1 \cap \alpha\inv(s_1)$ and $K_2 \cap \alpha\inv(s_2)$ respectively, we have $r_1 \in \opti{\Ds}{K_1 \cap \alpha\inv(s_1),\rho}$ and $r_2 \in \opti{\Ds}{K_2 \cap \alpha\inv(s_2),\rho}$. It follows that $(s_1,r_1) \in \lratauxd(K_1)$ and $(s_2,r_2) \in \lratauxd(K_2)$. Consequently, $(s,r_1r_2) = (s_1s_2,r_1r_2) \in \lratauxd(K_1) \cdot \lratauxd(K_2)$. Finally, by hypothesis and since $\rho$ is quasi-\tame, we have,
  \[
    r \leq \rho(H) = \rho(H_1H_2) =  \quasir(\rho(H_1) \cdot \rho(H_2)) = \quasir(r_1r_2).
  \]
  By definition of $\mu$, this yields$(s,r) \in \mu(\lratauxd(K_1) \cdot \lratauxd(K_2))$, concluding the proof.
\end{proof}

We present a second result for \ratms of the form $\bratauxbd: 2^{A^*} \to 2^R$. As expected, we shall consider the case $\Ds = \pol{\Cs}$. In that case, we are not able to prove that \bratauxd is quasi-\tame (the issue being that $\bool{\Ds}$ is \emph{not} closed under concatenation in general). We deal with this problem using two separate results. First, we show that when \Ds is a \pvari of regular languages \emph{closed under concatenation}, while \bratauxbd might not be quasi-\tame itself, it coincides with a quasi-\mratm over languages in \Ds. Note that this result is where we use the weak concatenation principle that we presented for Boolean closure in Lemma~\ref{lem:hintro:boolconcat}.

\begin{lemma} \label{lem:nest:bpolrat}
	Let \Ds be a \pvari closed under concatenation and $\rho: 2 ^{A^*} \to R$ a \mratm. There exists a \ratm $\tau: 2^{A^*} \to 2^R$ which satisfies the two following conditions:
	\begin{enumerate}
		\item $\tau$ is quasi-\tame for the endomorphism $\quasit: U\mapsto\dclosr U$ of $(2^R,\cup)$, and, 
		\item for every $K \in \Ds$, we have $\tau(K) = \bratauxbd(K)$.
	\end{enumerate}
\end{lemma}

\begin{proof}
	We first define $\tau: 2^{A^*} \to 2^R$. For every $K \subseteq A^*$, we let,
	\[
	\tau(K) = \bigcap_{\{L \in \Ds \mid K \subseteq L\}} \bratauxbd(L)
	\]
	It is immediate by definition that when $K \in \Ds$, we have $\tau(K) = \bratauxbd(K)$. We have to verify that $\tau$ is a quasi-\mratm. We need the following fact.

	\begin{fct} \label{fct:nest:bpolrat}
		For all $K \subseteq A^*$, there exists $L \in \Ds$ such that $K \subseteq L$ and $\tau(K) = \bratauxbd(L)$.
	\end{fct}

	\begin{proof}
		Since $2^R$ is finite, there are finitely many languages $L_1,\dots,L_n \in \Ds$ such that \mbox{$K \subseteq L_i$} for every $i \leq n$ and $\tau(K)=\bratauxbd(L_1) \cap \cdots \cap \bratauxbd(L_n)$. Since \Ds is a lattice by hypothesis, we have $L = L_1 \cap \cdots \cap L_n \in \Ds$. Clearly, $K \subseteq L$ which implies $\tau(K) \subseteq \bratauxbd(L)$ by definition of $\tau$. Finally, since $L \subseteq L_i$ for every $i$, we obtain that $\bratauxbd(L) \subseteq \bratauxbd(L_1) \cap \cdots \cap \bratauxbd(L_n) = \tau(K)$. Altogether, we conclude that $\tau(K) = \bratauxbd(L)$.
	\end{proof}

	Let us first prove that $\tau$ is a \ratm. Clearly, $\tau(\emptyset) = \bratauxbd(\emptyset) = \emptyset$. Let $K_1,K_2 \subseteq A^*$, we prove that $\tau(K_1\cup K_2)=\tau(K_1)\cup\tau(K_2)$. Clearly, every $H \in \Ds$ containing $K_1 \cup K_2$ contains $K_1$ and $K_2$ as well. Thus, $\tau(K_1) \cup \tau(K_2) \subseteq \tau(K_1 \cup K_2)$ by definition of $\tau$. For the converse inclusion, Fact~\ref{fct:nest:bpolrat} yields $L_1,L_2 \in \Ds$ such that $K_i \subseteq L_i$ and $\tau(K_i) = \bratauxbd(L_i)$ for $i=1,2$. Since \bratauxbd is a \ratm by Proposition~\ref{prop:areratms}, we get $\tau(K_1) \cup \tau(K_2) = \bratauxbd(L_1 \cup L_2)$. Since $K_1 \cup K_2 \subseteq L_1 \cup L_2$, the definition of $\tau$ then yields $\tau(K_1 \cup K_2) \subseteq \tau(K_1) \cup \tau(K_2)$.

	It remains to prove that $\tau$ is quasi-\tame for the endomorphism $\quasit: U \mapsto \dclosr U$ of $(2^R,\cup)$. Axiom~\ref{itm:nest:qt1} is immediate: for every $T,U,V \in  2^{R}$, we have $\dclosr (T (\dclosr U) V) = \dclosr (TUV)$.

	We turn to Axiom~\ref{itm:nest:qt3}. For $K_1,K_2 \subseteq A^*$, we prove that $\tau(K_1K_2) = \dclosr (\tau(K_1) \cdot \tau(K_2))$. We start with the right to left inclusion. Let $r \in \dclosr (\tau(K_1) \cdot \tau(K_2))$. We show that $r \in \tau(K_1K_2)$. By definition, this boils down to proving that $r \in \bratauxbd(H)$ for every $H \in \Ds$ such that $K_1K_2 \subseteq H$. We fix $H$ for the proof. Consider the two following languages,
	\[
	U_1 = \bigcap_{v \in K_2} Hv\inv \quad \text{and} \quad U_2 = \bigcap_{u \in U_1} u\inv H.
	\]
	Since \Ds is a \pvari of regular languages and $H \in \Ds$, we have $U_1,U_2 \in \Ds$ (recall that a regular language has finitely many quotients by the Myhill-Nerode theorem). Since $K_1K_2 \subseteq H$, one may verify that $K_i \subseteq U_i$ for $i = 1,2$. This yields $\tau(K_i) \subseteq \bratauxbd(U_i)$ by definition of $\tau$. Therefore, since we have $r \in \dclosr (\tau(K_1) \cdot \tau(K_2))$, we obtain $r_i \in \bratauxbd(U_i)$ (\emph{i.e.} $r_i \in \opti{\bool{\Ds}}{U_i,\rho}$) for $i = 1,2$ such that $r \leq r_1r_2$. By Lemma~\ref{lem:boolclos}, \bool{\Ds} is a \vari. Thus, since $\rho$ is a \mratm,  Lemma~\ref{lem:closmult} yields $r_1r_2 \in \opti{\bool{\Ds}}{U_1U_2,\rho}$. Since $r \leq r_1r_2$, we obtain $r \in \opti{\bool{\Ds}}{U_1U_2,\rho}$ by definition of \imprints. Finally, we have $U_1U_2 \subseteq H$ by definition of $U_2$. Hence, Fact~\ref{fct:inclus2} yields $r \in \opti{\bool{\Ds}}{H,\rho} = \bratauxbd(H)$.

	We turn to the converse inclusion. Let $r \in \tau(K_1K_2)$. We show that $r \in \dclosr (\tau(K_1) \cdot \tau(K_2))$. For $i = 1,2$, Fact~\ref{fct:nest:bpolrat} yields $L_i \in \Ds$ such that $K_i \subseteq L_i$ and $\tau(K_i) = \bratauxbd(L_i)$. We let $\Hb_i$ as an optimal \bool{\Ds}-cover (for $\rho$) of $L_i$. By definition, $\prin{\rho}{\Hb_i} = \bratauxbd(L_i) = \tau(K_i)$. Since $L_1,L_2 \in \Ds$, we may apply Lemma~\ref{lem:hintro:boolconcat} to build a \bool{\Ds}-cover \Hb of $L_1L_2$ such that for every $H \in \Hb$, there exist $H_1 \in \Hb_1$ and $H_2 \in \Hb_2$ such that $H \subseteq H_1H_2$. By hypothesis $r \in \tau(K_1K_2)$. Moreover, we have $K_1K_2 \subseteq L_1L_2$ and $L_1L_2 \in \Ds$ since \Ds is closed under concatenation. Thus, $r \in \bratauxbd(L_1L_2)$ by definition of $\tau$. It follows that $r \in \opti{\bool{\Ds}}{L_1L_2,\rho}$ and since \Hb is a \bool{\Ds}-cover of $L_1L_2$, we get $H \in \Hb$ such that $r \leq \rho(H)$. By definition of \Hb, there exist $H_1 \in \Hb_1$ and $H_2 \in \Hb_2$ such that $H \subseteq H_1H_2$. This implies that $r \leq \rho(H_1) \cdot \rho(H_2)$. Finally, $\rho(H_i) \in \tau(K_i)$ for $i=1,2$ by definition of $\Hb_i$. Thus, we get $r \in \dclosr (\tau(K_1) \cdot \tau(K_2))$ which concludes the proof.
\end{proof}

We complete Lemma~\ref{lem:nest:bpolrat} with another statement which requires that \Ds is closed under \emph{marked} concatenation (this will be the case in practice since we apply these results for $\Ds = \pol{\Cs}$). In this case as well, we use our weak concatenation principle for classes of the form \bool{\Ds}. Specifically, we apply the variant for marked concatenation (\emph{i.e.} Lemma~\ref{lem:hintro:boolmconcat}).

\begin{lemma} \label{lem:nest:bpolm}
	Let \Ds be a \pvari closed under marked concatenation and $\rho: 2 ^{A^*} \to R$ a \mratm. For every $K_1,K_2 \in\Ds$ and $a\in A$, we have,
	\[
	\bratauxbd(K_1aK_2) = \dclosr (\bratauxbd(K_1) \cdot \{\rho(a)\} \cdot \bratauxbd(K_2))
	\]
\end{lemma}

\begin{proof}
	For the sake of avoiding clutter, we write $\xi$ for $\bratauxbd$. We start with the right to left inclusion. Let $r \in \dclosr (\xi(K_1) \cdot  \{\rho(a)\} \cdot\xi(K_2))$. We show that $r\in\xi(K_1aK_2)$. By definition, we have $r_i \in \xi(K_i)$ for $i = 1,2$ such that $r \leq r_1\rho(a)r_2$. We have $r_i \in \opti{\bool{\Ds}}{K_i,\rho}$ by definition of $\xi=\bratauxbd$. Clearly, we also have $\rho(a) \in \opti{\bool{\Ds}}{\{a\},\rho}$. By Lemma~\ref{lem:boolclos}, \bool{\Ds} is a \vari. Thus, since $\rho$ is a \mratm,  Lemma~\ref{lem:closmult} yields $r_1\rho(a)r_2 \in \opti{\bool{\Ds}}{K_1aK_2,\rho}$. Since $r \leq r_1\rho(a)r_2$, we get $r \in \opti{\bool{\Ds}}{K_1aK_2,\rho}$. By definition, this exactly says that $r\in\xi(K_1aK_2)$. We turn to the converse inclusion.

	Let $r \in \xi(K_1aK_2)$. We show that $r \in \dclosr (\xi(K_1) \cdot  \{\rho(a)\} \cdot \xi(K_2))$. For $i=1,2$, we let $\Hb_i$ as an optimal \bool{\Ds}-cover (for $\rho$) of $K_i$ for $i = 1,2$. By definition, we have $\prin{\rho}{\Hb_i} = \xi(K_i)$. Since $K_1,K_2 \in \Ds$, we may apply Lemma~\ref{lem:hintro:boolmconcat} to build a \bool{\Ds}-cover \Hb of $K_1aK_2$ such that for every $H \in \Hb$, there exist $H_1 \in \Hb_1$ and $H_2 \in \Hb_2$ such that $H \subseteq H_1aH_2$. By hypothesis, $r \in \xi(K_1aK_2) = \opti{\bool{\Ds}}{K_1aK_2,\rho}$. Therefore, since \Hb is a \bool{\Ds}-cover of $K_1aK_2$, we have $r \leq \rho(H)$ for some $H \in \Hb$. The definition of \Hb then yields $H_1 \in \Hb_1$ and $H_2 \in \Hb_2$ such that $H \subseteq H_1aH_2$. This implies that $r \leq \rho(H_1) \cdot \rho(a) \cdot \rho(H_2)$. Finally, $\rho(H_i) \in \xi(K_i)$ for $i=1,2$ by definition of $\Hb_i$. This yields $r \in \dclosr (\xi(K_1) \cdot  \{\rho(a)\} \cdot \xi(K_2))$ which concludes the proof.
\end{proof}

\section{Characterization of \pol{\Cs}-optimal \imprints}
\label{sec:pol}
In this section, we recall the theorem of~\cite{pseps3j} for classes of the form \pol{\Cs} (when \Cs is a finite \vari). It states a characterization of \pol{\Cs}-optimal \imprints: for a \Cs-compatible morphism $\alpha: A^* \to M$ and a \mratm $\rho: 2^{A^*} \to R$, \pocopti is characterized as the least subset of $M \times R$ satisfying specific properties. When used in the special case when $\rho$ is \nice, one obtains a least fixpoint procedure for computing \pocopti from $\alpha$ and $\rho$. By Proposition~\ref{prop:thereductionlat}, this yields an algorithm for solving \pol{\Cs}-covering. However, deciding \pol{\Cs}-covering is not our motivation here: we need this characterization in order to use it as a subresult when proving Theorem~\ref{thm:bpol:carac}.

Unfortunately, there are technical complications. When we apply the characterization as a subresult, we shall do so for \ratms which are not \tame, only quasi-\tame (as expected, we build them using the constructions presented in Section~\ref{sec:tools}). This case is not covered by the statement of~\cite{pseps3j}, which only deals with (true) \mratms. Consequently, we have to generalize this statement. We avoid redoing the whole proof, by obtaining the generalized statement as a corollary of the original one from~\cite{pseps3j}. Additionally, we take this opportunity to slightly tweak the original statement in order to better accommodate our use of the theorem in Sections~\ref{sec:sound} and~\ref{sec:comp}.

We first present the theorem and then focus on its proof. We fix an arbitrary finite \vari \Cs for the presentation.

\subsection{Statement}

Consider a \Cs-compatible morphism $\alpha: A^* \to M$ and a \ratm $\rho: 2^{A^*} \to R$ which is quasi-\tame (there is no other constraint on $\rho$, in particular, it need not be \nice). Recall that since $\alpha$ is \Cs-compatible, we know that for every $s \in M$, $\ctype{s}$ is well-defined as a \equc-class containing $\alpha\inv(s)$. We say that a subset $S \subseteq M \times R$ is \emph{\pol{\Cs}-saturated} (for $\alpha$ and $\rho$) when it satisfies the following properties:
\begin{enumerate}
\item \emph{Trivial elements:} For every $w \in A^*$, $(\alpha(w),\rho(w)) \in S$.
\item \emph{Downset:} We have $S = \dclosr S$.
\item \emph{Multiplication}: For every $(s_1,r_1),(s_2,r_2) \in S$, we have $(s_1s_2,r_1r_2) \in S$.
\item\label{op:half:polclos} \emph{\pol{\Cs}-closure}: For every pair of (multiplicative) idempotents $(e,f) \in S$, we have:
  \[
    (e,f \cdot \rho(\ctype{e}) \cdot f) \in S.
  \]
\end{enumerate}

We prove the following statement as a corollary of~\cite[Theorem~6.5]{pseps3j}.

\begin{theorem}\label{thm:half:mainpolc}
  Consider a \Cs-compatible morphism $\alpha: A^* \to M$ and a quasi-\mratm $\rho:2^{A^*} \to R$. Moreover, let $S$ be the least \pol{\Cs}-saturated subset of $M \times R$. Then,
  \[
    \pocopti = \dclosr \big\{(s,\quasir(r)) \mid (s,r) \in S\big\}.
  \]
\end{theorem}

Recall that a (true) \mratm $\rho: 2^{A^*} \to R$ is also quasi-\tame for the  endomorphism \quasir defined as the identity on $R$ (see Remark~\ref{rem:tameisquasi}). Thus, in this case, Theorem~\ref{thm:half:mainpolc} yields that $\pocopti = \dclosr S = S$ where $S$ is the least \pol{\Cs}-saturated subset of $M \times R$. This is the original statement of~\cite{pseps3j}. When $\rho$ is a \nice \mratm, it is clear that one may compute the least \pol{\Cs}-saturated subset of $M \times R$ with a least fixpoint algorithm. Therefore, we get an algorithm for \pol{\Cs}-covering by Proposition~\ref{prop:thereductionlat}.

\begin{remark}
  There is a difference between Theorem~\ref{thm:half:mainpolc} for \pol{\Cs} and Theorem~\ref{thm:bpol:carac} for \bpol{\Cs}: the latter is restricted to {\bf \nice} \ratms while this is not the case for the former. As explained in Remark~\ref{rem:niceiscrucial}, while it is easy to miss, this difference is crucial.
\end{remark}

\subsection{Proof of Theorem~\ref{thm:half:mainpolc}}

We fix a \Cs-compatible morphism $\alpha: A^* \to M$ and a quasi-\tame \ratm $\rho: 2^{A^*} \to R$. Since $\rho$ is quasi-\tame, we have a semiring structure $(R,+,\cdot)$ on $R$ and an endomorphism \quasir of $(R,+)$ satisfying the appropriate axioms. Finally, we let $S$ as the least \pol{\Cs}-saturated subset of $M \times R$ for $\alpha$ and $\rho$. We show that $\pocopti = \dclosr \big\{(s,\quasir(r)) \mid (s,r) \in S\big\}$.

We first prove that the surjective restriction of $\rho$ is a \emph{true} \mratm (for a new multiplication on the rating set which is distinct from ``$\cdot$''). This allows us to apply the theorem of~\cite{pseps3j}.

\smallskip

We define a new multiplication on $R$ that we denote by ``$\odot$''. For every $q,r \in R$, we define $q \odot r = \quasir(qr)$. It is immediate from Axiom~\ref{itm:nest:qt1} in the definition of quasi-\mratms that ``$\odot$'' is associative. Moreover, since \quasir is an endomorphism of $(R,+)$, one may verify that ``$\odot$'' distributes over addition and that the element $0_R$ is a zero for ``$\odot$''. Yet, note that $(R,+,\odot)$ need not be a semiring: ``$\odot$'' might not have a neutral element.

Axiom~\ref{itm:nest:qt3} in the definition of quasi-\mratms solves this issue. It implies that for $K_1,K_2 \subseteq A^*$, we have $\rho(K_1K_2)  = \rho(K_1) \odot \rho(K_2)$. Thus, the rating map $\rho$ is a semigroup morphism from $(2^{A^*},\cdot)$ to $(R,\cdot)$. Let $Q = \rho(2^{A^*}) \subseteq R$ and let $\tau: 2^{A^*} \to Q$ be the surjective restriction of $\rho$. It follows that $(Q,+,\odot)$ is an idempotent semiring (the neutral element is $\rho(\veps) = \tau(\veps) \in Q$) and $\tau: (2^{A^*},\cup,\cdot) \to (Q,+,\odot)$ is a semiring morphism, \emph{i.e.} a \mratm. Additionally, since we defined $\tau$ as the surjective restriction of $\rho$, it is immediate that,
\[
\pocopti = \dclosr \left(\popti{\pol{\Cs}}{\alpha}{\tau}\right)
\]
Moreover, since $\tau$ is a true \mratm, we may apply the theorem of~\cite{pseps3j}. Let $T$ be the least \pol{\Cs}-saturated of $M \times Q$ for $\alpha$ and $\tau$. We obtain from~\cite[Theorem~6.5]{pseps3j} that $\popti{\pol{\Cs}}{\alpha}{\tau} = T$. Hence, it now suffices to prove that,
\begin{equation} \label{eq:nest:polc}
\dclosr T = \dclosr \{(s,\quasir(r)) \mid (s,r) \in S\}.
\end{equation}
Indeed, this clearly implies $\pocopti = \dclosr \big\{(s,\quasir(r)) \mid (s,r) \in S\big\}$ which concludes the proof of Theorem~\ref{thm:half:mainpolc}.

\begin{remark} \label{rem:nest:odot}
	We are dealing with two strongly connected \pol{\Cs}-saturated sets. Namely, $S \subseteq M \times R$ and $T \subseteq M \times Q \subseteq M \times R$. However, there is a subtle difference between the two. By definition, $S$ is \pol{\Cs}-saturated for $\alpha$ and $\rho$. This notion depends on the original multiplication ``$\cdot$'' of $R$. On the other hand, $T$ is \pol{\Cs}-saturated for $\alpha$ and $\tau$. By definition of $\tau$, this notion depends on the new multiplication ``$\odot$'' of $Q$.
\end{remark}

It now remains to prove~\eqref{eq:nest:polc}. We handle the two inclusions separately. First, we show that $\dclosr T \subseteq \dclosr \{(s,\quasir(r)) \mid (s,r) \in S\}$. This boils down to proving that for every $(s,q) \in T$, we have $r \in R$ such that $(s,r) \in S$ and $q \leq \quasir(r)$. By definition, $T$ is the least \pol{\Cs}-saturated subset of $M \times Q$ for $\alpha$ and $\tau$. Therefore, every pair $(s,q) \in T$ is built from trivial elements using downset, multiplication and \pol{\Cs}-closure (here, the multiplication is ``$\odot$'' on $Q$, see Remark~\ref{rem:nest:odot}). We proceed by induction on this construction.

If $(s,q)$ is a trivial element, we have $w\in A^*$ such that $s = \alpha(w)$ and $q = \tau(w) = \rho(w)$. We have $(s,q) \in S$ since $S$ is \pol{\Cs}-saturated for $\alpha$ and $\rho$. Moreover, $q = \quasir(q)$ by Fact~\ref{fct:nest:quasip} since $q = \rho(w)$ which concludes this case. We turn to downset: we have $q' \in R$ such that $(s,q') \in T$ and $q \leq q'$. Induction yields $r \in R$ such that $(s,r) \in S$ and $q' \leq \quasir(r)$. Thus, $q \leq \quasir(r)$ which concludes this case. We turn to multiplication. In that case, we have $(s_1,q_1),(s_2,q_2) \in T$ such that $s=s_1s_2$ and $q = q_1 \odot q_2$. For $i=1,2$, induction yields $r_i \in R$ such that $(s_i,r_i) \in S$ and $q_i \leq \quasir(r_i)$. Since $S$ is \pol{\Cs}-saturated for $\alpha$ and $\rho$, we obtain $(s,r_1r_2) \in S$. Moreover, Axiom~\ref{itm:nest:qt1} in the definition of quasi-\mratms yields,
\[
q = q_1 \odot q_2 = \quasir(q_1q_2) \leq  \quasir(\quasir(r_1)\quasir(r_2)) = \quasir(r_1r_2)
\]
This concludes the proof for this case. It remains to handle \pol{\Cs}-closure. In that case, we have a pair of idempotents $(e,f) \in T$ such that $s = e$ and $q = f \odot \tau(\ctype{e}) \odot f$ (here, $f$ is an idempotent of $(R,\odot)$ since $T$ is \pol{\Cs}-saturated for $\alpha$ and $\tau$). Induction yields $r \in R$ such that $(e,r) \in S$ and $f \leq \quasir(r)$. Since $(R,\cdot)$ is a finite monoid, there exists a number $p \geq 1$ such that $r^p$ is an idempotent of $(R,\cdot)$. Since $S$ is \pol{\Cs}-saturated for $\alpha$ and $\rho$, we obtain from closure under multiplication that $(e,r^p) \in S$. Together with \pol{\Cs}-closure, this yields $(e,r^p\rho(\ctype{e})r^p) \in S$. Since $f \leq \quasir(r)$ and $f$ is an idempotent of $(R,\odot)$, we obtain,
\[
f \odot \tau(\ctype{e}) \odot f = f \odot \rho(\ctype{e}) \odot f \leq \quasir(r^p\rho(\ctype{e})r^p)
\]
This concludes the proof for the inclusion $\dclosr T \subseteq \dclosr \{(s,\quasir(r)) \mid (s,r) \in S\}$ in~\eqref{eq:nest:polc}.

\medskip

We turn to the converse inclusion: $\dclosr \{(s,\quasir(r)) \mid (s,r) \in S\} \subseteq \dclosr T$. It suffices to show that for every $(s,r) \in S$, we have $(s,\quasir(r)) \in \dclosr T$. By definition, $S$ is the least \pol{\Cs}-saturated subset of $M \times R$ for $\alpha$ and $\rho$. Hence, every pair $(s,r) \in S$ is built from trivial elements using downset, multiplication and \pol{\Cs}-closure (here, the multiplication is ``$\cdot$'' on $R$, see Remark~\ref{rem:nest:odot}). We proceed by induction on this construction.

If $(s,r)$ is a trivial element, we have $w\in A^*$ such that $s = \alpha(w)$ and $r = \rho(w) = \tau(w)$. By Fact~\ref{fct:nest:quasip}, we have $\quasir(r) = \rho(w)$. Thus, $(s,\quasir(r)) \in T \subseteq \dclosr T$ since $T$ is \pol{\Cs}-saturated for $\alpha$ and $\tau$. We turn to downset: we have $r' \in R$ such that $(s,r') \in T$ and $r \leq r'$. Induction yields $(s,\quasir(r')) \in \dclosr T$. Thus, since $\quasir(r) \leq \quasir(r')$, we get $(s,\quasir(r)) \in \dclosr \dclosr T = \dclosr T$ as desired. We turn to multiplication. In that case, we have $(s_1,r_1),(s_2,r_2) \in S$ such that $s= s_1s_2$ and $r = r_1r_2$. By induction, we obtain that $(s_i,\quasir(r_i))\in \dclosr T$ for $i=1,2$. This yields $q_1, q_2 \in Q$ such that  $(s_i,q_i) \in T$ and $\quasir(r_i)\leq q_i$ for $i=1,2$. Since $T$ is \pol{\Cs}-saturated for $\alpha$ and $\tau$, it follows that $(s,q_1 \odot q_2) \in T$ by closure under multiplication. Moreover,
\[
\quasir(r) = \quasir(r_1r_2) = \quasir(\quasir(r_1)\quasir(r_2)) \leq \quasir(q_1q_2) = q_1 \odot q_2
\]
Altogether, we obtain $(s,\quasir(r))\in\dclosr T$ as desired. It remains to handle \pol{\Cs}-closure. We have a pair of multiplicative idempotents $(e,f) \in S$ such that $s = e$ and $r = f \cdot \rho(\ctype{e}) \cdot f$ (here, $f$ is an idempotent of $(R,\cdot)$ since $S$ is \pol{\Cs}-saturated for $\alpha$ and $\rho$). By induction, we have $(e,\quasir(f)) \in \dclosr T$. Thus, we obtain $q \in Q$ such that $(e,q) \in T$ and $\quasir(f) \leq q$. Since $(R,\odot)$ is a finite semigroup, there exists a number $p \geq 1$ such the multiplication of $p$ copies of $q$ with ``$\odot$'' is an idempotent of $(R,\odot)$. We write $g \in R$ for this idempotent. Since $T$ is \pol{\Cs}-saturated for $\alpha$ and $\tau$, we obtain from closure under multiplication that $(e,g) \in T$. Together with \pol{\Cs}-closure, this yields $(e,g \odot \tau(\ctype{e}) \odot g) \in T$. Since $\quasir(f) \leq q$ and $f$ is an idempotent of $(R,\cdot)$, one may verify that,
\[
\quasir(f \cdot \rho(\ctype{e}) \cdot f) = \quasir(f \cdot \tau(\ctype{e}) \cdot f) \leq g \odot \tau(\ctype{e}) \odot g
\]
Altogether, we obtain $(s,\quasir(r)) = (e,\quasir(f \cdot \rho(\ctype{e}) \cdot f)) \in \dclosr T$ as desired. This concludes the proof.

\section{Soundness in Theorem~\ref{thm:bpol:carac}}
\label{sec:sound}
We may now start the proof of Theorem~\ref{thm:bpol:carac}. In this section, we establish that the statement is sound. The proof is divided in two parts. First, we present a preliminary result, which applies to the Boolean closure operation in general, \emph{i.e.}, to classes of the form $\bool{\Ds}$ when \Ds is an arbitrary lattice. Then, we apply this preliminary result in the special case when $\Ds = \pol{\Cs}$ (for \Cs a finite \vari) to establish the soundness direction in Theorem~\ref{thm:bpol:carac}.

\subsection{Preliminary result}

We first introduce terminology that we need to state our result. We fix a lattice \Ds. Moreover, we let $\rho: 2^{A^*} \to R$ be a \ratm. Let us recall the definition of the \ratm \lratauxd, defined page~\pageref{rataux}, which also depends on a map $\alpha: A^* \to M$:
\[
  \begin{array}{llll}
    \lratauxd: & (2^{A^*},{\cup}) & \to & (2^{M \times R},\cup)\\
               & K & \mapsto & \big\{(s,r) \mid r \in \opti{\Ds}{\alpha\inv(s) \cap K,\rho}\big\}.
  \end{array}
\]

Using induction, we define a \ratm $\tau_n: 2^{A^*} \to Q_n$ for every $n \in \nat$. When $n = 0$, the rating set $Q_0$ is $(2^R,\cup)$ and $\tau_0$ is defined as follows, 
\[
  \begin{array}{llll}
    \tau_0: & (2^{A^*},\cup) & \to     & (2^R,\cup)               \\
         & K              & \mapsto & \{\rho(w) \mid w \in K\}.
  \end{array}
\]
It is immediate by definition that $\tau_0$ is indeed a \ratm (\emph{i.e.}, a monoid morphism).

Assume now that $n \geq 1$ and that $\tau_{n-1}: 2^{A^*} \to Q_{n-1}$ is defined. Recall that $\rho_*:  A^* \to R$ denotes the canonical map associated to the \ratm $\rho$. We define $\tau_n$ as:
\[
  \tau_n= \lrataux{\Ds}{\rho_*}{\tau_{n-1}}.
\]
By Proposition~\ref{prop:areratms}, $\tau_n$ is indeed a \ratm. By definition, this means that for all $n \geq 1$, the rating set $Q_n$ of $\tau_n$ is
\[
  Q_n=(2^{R \times Q_{n-1}},\cup).
\]
We complete this definition with maps $f_n: Q_n \to 2^R$ for $n \in \nat$, defined by induction on~$n$.
\begin{itemize}
\item For $n = 0$, let $T \in Q_0= 2^R$. We define:
  \[
    f_0(T) =  \dclosr \{r_1 + \cdots + r_k \mid r_1,\dots,r_k \in T\}.
  \]
\item For $n \geq 1$, let $T \in Q_n = 2^{R \times Q_{n-1}}$. We define:
  \[
    f_n(T) = {\dclosr} \left\{r_1 + \cdots + r_k \mid
      \begin{array}{c}
        \text{there exist $(r_1,T_1),\dots,(r_k,T_k) \in T$ such that}\\
        \text{$r_1 + \cdots + r_k \in f_{n-1}(T_i)$ for every $i \leq k$}
      \end{array}
    \right\}.
  \]
\end{itemize}
The following fact is immediate from the definition.

\begin{fct} \label{fct:fmapinc}
  For every $n \in \nat$ and $U,U' \in Q_n$ such that $U \subseteq U'$, we have $f_n(U) \subseteq f_n(U')$.
\end{fct}

We may now state the preliminary result that we shall use in our soundness direction of Theorem~\ref{thm:bpol:carac}.

\begin{proposition} \label{prop:sound1}
  Consider a language $L \in \Ds$. Then, the following inclusion holds:
  \[
    \bigcap_{n \in \nat} f_n(\tau_n(L)) \subseteq \opti{\bool{\Ds}}{L,\rho}.
  \]
\end{proposition}

\begin{proof}
  The proof is based on the following more involved statement which is proved by induction on $n \in \nat$.

  \begin{lemma} \label{lem:soundreal}
    Let $n \in \nat$ and $L,K_0,\dots,K_n,H_0,\dots,H_n \in \Ds$ such that $\{K_i \setminus H_i \mid i \leq n\}$ is a cover of $L$. Then, for every $s \in f_{2n}(\tau_{2n}(L))$, there exists $j \leq n$ such that $s \leq \rho(K_j \setminus H_j)$.
  \end{lemma}

  Before proving the lemma, let us use it to prove the first property described in Proposition~\ref{prop:sound1}. Let $L \in \Ds$ be a language. We write $S$ for the set,
  \[
    S = \bigcap_{n \in \nat} f_n(\tau_n(L)).
  \]
  We show that $S \subseteq \opti{\bool{\Ds}}{L,\rho}$. First, we prove the following fact which describes a special optimal \bool{\Ds}-cover of $L$ for $\rho$.

  \begin{fct} \label{fct:bpol:refine}
    There exist $n \in \nat$ and $K_0,\dots,K_n,H_0,\dots,H_n \in \Ds$ such that $\{K_i \setminus H_i \mid i \leq n\}$ is an optimal \bool{\Ds}-cover of $L$ for $\rho$.
  \end{fct}

  \begin{proof}
    Let \Hb be an arbitrary optimal \bool{\Ds}-cover of $L$ for $\rho$. Each $V \in \Hb$ is the Boolean combination of languages in \Ds. We put it in disjunctive normal form. Each disjunct is an intersection languages belonging to \Ds, or whose complement belongs to \Ds. Since \Ds is lattice, both \Ds and the complement class \cocl{\Ds} are closed under intersection. Therefore, each disjunct in the disjunctive normal form of $V$ is actually of the form $K \setminus H$, where $K,H$ both belong to \Ds. We let \Kb as the set of all languages $K \setminus H$ which are a disjunct in the disjunctive normal form of some $V \in \Hb$. Clearly, \Kb remains a \bool{\Ds}-cover of $L$ since \Hb was one. Moreover, it is immediate that $\prin{\rho}{\Kb} \subseteq \prin{\rho}{\Hb}$ since every language in \Kb is included in a language of \Hb. Hence, \Kb remains an optimal \bool{\Ds}-cover of $L$ for $\rho$ since \Hb was one.
  \end{proof}

  We let $n \in \nat$ and $K_0,\dots,K_n,H_0,\dots,H_n \in \Ds$ be as defined in Fact~\ref{fct:bpol:refine}. We may now prove that $S \subseteq \opti{\bool{\Ds}}{L,\rho}$. Let $s \in S$. By hypothesis on $S$, we have $s \in f_{2n}(\tau_{2n}(L))$. Therefore, since $\{K_i \setminus H_i \mid i \leq n\}$ is by definition a cover of $L$, it is immediate from Lemma~\ref{lem:soundreal} that there exists $j \leq n$ such that $s \leq \rho(K_j \setminus H_j)$. Since $\{K_i \setminus H_i \mid i \leq n\}$ is an optimal \bool{\Ds}-cover of $L$ for $\rho$, this implies that $s \in \opti{\bool{\Ds}}{L,\rho}$ which concludes the main proof.

  \medskip

  We turn to the proof of Lemma~\ref{lem:soundreal}. The argument is an induction on $n \in \nat$. We start with the base case $n = 0$.

  \medskip
  \noindent
  {\bf Base case.} Consider $L,K_0,H_0 \in \Ds$ such that $\{K_0 \setminus H_0\}$ is a cover of $L$ and let $s \in f_{0}(\tau_{0}(L))$. We have to show that $s \leq \rho(K_0 \setminus H_0)$. By definition of $f_0$, we get $r_1,\dots,r_k \in \tau_{0}(L)$ such that $s \leq r_1 + \cdots + r_k$. Moreover, the definition of $\tau_0$ yields that for every $i \leq k$, $r_i = \rho(w_i)$ for some $w_i \in L$. Therefore, $r_i \leq \rho(L)$ for every $i \leq k$ and since $R$ is idempotent for addition, $s \leq r_1 + \cdots + r_k \leq \rho(L)$. Finally, since $\{K_0 \setminus H_0\}$ is a cover of $L$, we have $L \subseteq K_0 \setminus H_0$ and we get $s \leq \rho(K_0 \setminus H_0)$, finishing the argument for the base case.

  \medskip
  \noindent
  {\bf Inductive step.} We now assume that $n \geq 1$.  Let $L,K_0,\dots,K_n,H_0,\dots,H_n \in \Ds$ such that $\{K_i \setminus H_i \mid i \leq n\}$ is a cover of $L$ and let  $s \in f_{2n}(\tau_{2n}(L))$. We have to exhibit $j \leq n$ such that $s \leq \rho(K_j \setminus H_j)$. Using the hypothesis that $s \in f_{2n}(\tau_{2n}(L))$, we prove the following fact.

  \begin{fct} \label{fct:tautau}
    There exists $(r,U) \in \tau_{2n}(L)$ such that $s \in f_{2n-1}(U)$.
  \end{fct}

  \begin{proof}
    By definition of $f_{2n}$, the hypothesis that $s \in f_{2n}(\tau_{2n}(L))$ yields $(r'_1,U'_1),\dots,(r'_k,U'_k) \in \tau_{2n}(L)$ such that $r'_1 + \cdots + r'_k \in f_{2n-1}(U'_i)$ for every $i \leq k$ and $s \leq r'_1 + \cdots + r'_k$. We let $(r,U) = (r'_1,U'_1) \in \tau_{2n}(L)$. We have $r'_1 + \cdots + r'_k \in f_{2n-1}(U)$ and $s \leq r'_1 + \cdots + r'_k$. Hence, since $f_{2n-1}(U)$ is closed under downset (by definition), this implies $s \in f_{2n-1}(U)$.
  \end{proof}

  Recall that by definition, $\tau_{2n}$ is the \ratm \lrataux{\Ds}{\rho_*}{\tau_{2n-1}}. Hence, $(r,U) \in \tau_{2n}(L)$ means that:
  \[
    U \in \opti{\Ds}{\rho_*\inv(r) \cap L,\tau_{2n-1}},
  \]
  which yields, by Fact~\ref{fct:inclus2}, that,
  \[
    U \in \opti{\Ds}{L,\tau_{2n-1}}.
  \]
  Since $\{K_i \setminus H_i \mid i \leq n\}$ is a cover of $L$, $L,K_1,\dots,K_n \in \Ds$ and \Ds is a lattice, it follows that $\{L \cap K_i \mid i \leq n\}$ is a \Ds-cover of $L$. Therefore, since $U \in \opti{\Ds}{L,\tau_{2n-1}}$, we obtain some $\ell \leq n$ such that $U \subseteq \tau_{2n-1}(L \cap K_\ell)$.

  Furthermore, since $s \in f_{2n-1}(U)$, we may unravel the definition of $f_{2n-1}$ which yields $(r_1,U_1),\dots,(r_k,U_k) \in U$ such that $r_1 + \cdots + r_k \in f_{2(n-1)}(U_m)$ for every $m \leq k$ and $s \leq r_1 + \cdots + r_k$. Note that by definition of $f_{2(n-1)}$, this also implies that $s \in f_{2(n-1)}(U_m)$ for every $m \leq k$. We now distinguish two sub-cases.

  \medskip
  \noindent
  {\it Sub-case~1:} Assume that for every $m \leq k$, we have,
  \[
    \rho_*\inv(r_m) \cap (K_\ell \setminus H_\ell) \neq \emptyset
  \]
  This means that for every $m \leq k$, we have $w_m \in K_\ell \setminus H_\ell$ such that $\rho(w_m) = r_m$. In particular, $r_m = \rho(w_m) \leq  \rho(K_\ell \setminus H_\ell)$ for every $m \leq k$. Finally, since $R$ is idempotent for addition, we get $r_1 + \cdots + r_k \leq \rho(K_\ell \setminus H_\ell)$. Since $s \leq r_1 + \cdots + r_k$, we get $s \leq \rho(K_\ell \setminus H_\ell)$ and Lemma~\ref{lem:soundreal} holds for $j = \ell$ in this case.

  \medskip
  \noindent
  {\it Sub-case~2:} Conversely, assume that there exists $m \leq k$ such that,
  \[
    \rho_*\inv(r_m) \cap (K_\ell \setminus H_\ell) = \emptyset
  \]
  This implies that $K_\ell \cap \rho_*\inv(r_m) \subseteq K_\ell \cap H_\ell$. Recall that $(r_m,U_m) \in U$ and that $U \subseteq \tau_{2n-1}(L \cap K_\ell)$. Since $\tau_{2n-1}$ is the \ratm \lrataux{\Ds}{\rho_*}{\tau_{2(n-1)}} by definition, it follows that,
  \[
    U_m \in \opti{\Ds}{\rho_*\inv(r_m) \cap L \cap K_\ell,\tau_{2(n-1)}}
  \]
  Combined with the inclusion $K_\ell \cap \rho_*\inv(r_m) \subseteq K_\ell \cap H_\ell$ and Fact~\ref{fct:inclus2}, this yields that,
  \[
    U_m \in \opti{\Ds}{L \cap K_\ell \cap H_\ell,\tau_{2(n-1)}}
  \]
  Since \Ds is a lattice, it is clear that $\{L \cap K_\ell \cap H_\ell\}$ is a \Ds-cover of $L \cap K_\ell \cap H_\ell$. Thus, we obtain that $U_m \subseteq \tau_{2(n-1)}(L \cap K_\ell \cap H_\ell)$. Therefore, since $s \in f_{2(n-1)}(U_m)$ by definition of $U_m$, we obtain from Fact~\ref{fct:fmapinc} that,
  \[
    s \in f_{2(n-1)}(\tau_{2(n-1)}(L \cap K_\ell \cap H_\ell))
  \]
  Finally, since $\{K_i \setminus H_i \mid i \leq n\}$ was a cover of $L$, it is clear that $\{K_i \setminus H_i \mid i \leq n \text{ and } i \neq \ell\}$ (of size $n-1$) is a cover of $L \cap K_\ell \cap H_\ell$. Therefore, it follows by induction on $n$ in Lemma~\ref{lem:soundreal} that there exists $j \leq n$ (with $j \neq \ell$) such that $s \leq \rho(K_j \setminus H_j)$, finishing the proof.
\end{proof}

\subsection{Soundness proof for Theorem~\ref{thm:bpol:carac}}

We may now come back to our main objective: soundness in Theorem~\ref{thm:bpol:carac}. We fix a finite \vari \Cs and a \mratm $\rho: 2^{A^*} \to R$. We show that for every \bpol{\Cs}-saturated subset $S \subseteq (\sclac) \times R$ (for $\rho$), we have $S \subseteq \cbpopti$.

\begin{remark}
  Note that we do not use the hypothesis that $\rho$ is \nice for this direction. This is only needed for completeness.
\end{remark}

By Theorem~\ref{thm:polclos}, \pol{\Cs} is a lattice. Therefore, we may instantiate the definitions and results presented at the beginning of the section for our \nice \mratm $\rho: 2^{A^*} \to R$ in the special case when $\Ds = \pol{\Cs}$. We keep using the same notation: we have the \ratms $\tau_n: 2^{A^*} \to Q_n$ (as we prove below, they are now quasi-\tame since $\Ds = \pol{\Cs}$ and $\rho$ is \tame) and the maps $f_n: Q_n \to 2^R$.

\medskip

We complete Proposition~\ref{prop:sound1} with another statement specific to this special case. In fact, the proof of this second proposition is based on Theorem~\ref{thm:half:mainpolc}, the characterization of \pol{\Cs}-optimal \imprints (we apply it to the \ratms $\tau_n$). Together, these two results imply soundness in Theorem~\ref{thm:bpol:carac}.

\begin{proposition} \label{prop:sound2}
  Consider $S \subseteq (\sclac) \times R$ which is \bpol{\Cs}-saturated for $\rho$. Then, the following inclusion holds for every $D \in \sclac$,
  \[
    S(D) \subseteq \bigcap_{n \in \nat} f_n(\tau_n(D)).
  \]
\end{proposition}

When put together, Proposition~\ref{prop:sound1} and Proposition~\ref{prop:sound2} imply soundness in Theorem~\ref{thm:bpol:carac}.   Indeed, consider a \bpol{\Cs}-saturated set $S \subseteq (\sclac) \times R$. We show that $S \subseteq \cbpopti$. This amounts to proving that for every $D \in \sclac$, we have,
\[
  S(D) \subseteq \cbpopti(D) = \opti{\bpol{\Cs}}{D,\rho}.
\]
It is immediate from Proposition~\ref{prop:sound2} that,
\[
  S(D) \subseteq \bigcap_{n \in \nat} f_n(\tau_n(D)).
\]
Moreover, since $D \in \sclac$, we have $D \in \Cs$, which implies that $D \in \pol{\Cs}$. Therefore, Proposition~\ref{prop:sound1} yields that,
\[
  \bigcap_{n \in \nat} f_n(\tau_n(D)) \subseteq \opti{\bpol{\Cs}}{D,\rho}.
\]
Altogether, we get the desired inclusion: $S(D) \subseteq \opti{\bpol{\Cs}}{D,\rho}$. This concludes the soundness proof.

\medskip

It remains to prove Proposition~\ref{prop:sound2}. We start with a few additional results about the \ratms $\tau_n$ that we are able to prove using our new hypotheses (\emph{i.e.}, that $\rho$ is \tame and $\Ds = \pol{\Cs}$).

\medskip
\noindent
{\bf Preliminaries.} Recall that $\rho$ is \tame. Hence, $R$ is a semiring $(R,+,\cdot)$. Since $Q_0 = 2^R$ and $Q_n = 2^{R \times Q_{n-1}}$ for all $n \geq 1$, we may lift the multiplication of $R$ to all the rating sets $Q_n$ in the natural way. It is simple to verify that $(Q_n,\cup,\cdot)$ is a semiring for every $n \in \nat$. We first show that the \ratms $\tau_n$ are quasi-\tame for these multiplications.

\begin{lemma} \label{lem:squasitame}
  The \ratm $\tau_0: 2^{A^*} \to Q_0$ is \tame. Moreover, for every $n \geq 1$, the \ratm $\tau_n: 2^{A^*} \to Q_n$ is quasi-\tame for the following endomorphism $\quasi{\tau_n}$ of $(Q_n,\cup)$:
  \[
    \quasi{\tau_n}(T) = \dclosp{Q_{n-1}} \{(r,\quasi{\tau_{n-1}}(V)) \mid (r,V) \in T\} \quad \text{for every $T \in Q_n = 2^{R \times Q_{n-1}}$}.
  \]
\end{lemma}

\begin{proof}
  Recall that $Q_0 = 2^R$ and $\tau_0(K) = \{\rho(w) \mid w \in K\}$ for every language $K$. It is immediate that $\tau_0$ is \tame since $\rho$ is \tame by hypothesis. The result for the \ratms $\tau_n$ for $n \geq 1$ is then immediate from Lemma~\ref{lem:copelrat} using induction on $n$ since \pol{\Cs} is a \pvari closed under concatenation by Theorem~\ref{thm:polclos}.
\end{proof}

We complete this result with two lemmas which connect the hypothesis that the \ratms $\tau_n$ are quasi-\tame with the maps $f_n: Q_n \to 2^R$.

\begin{lemma} \label{lem:fmapmu}
  For every $n \in \nat$ and $T \in Q_n$, we have $f_n(T) \subseteq f_n(\quasi{\tau_n}(T))$.
\end{lemma}

\begin{proof}
  We proceed by induction on $n \in \nat$. When $n = 0$, $\tau_0$ is \tame and the endomorphism \quasi{\tau_0} of $(Q_0,\cup)$ is the identity on $Q_0$. Hence, the lemma is immediate. Assume now that $n \geq 1$. Consider $T \in Q_n$. Let $r \in f_n(T) \subseteq R$. By definition, we have $(r_1,T_1),\dots,(r_k,T_k) \in T$ such that $r_1 + \cdots + r_k \in f_{n-1}(T_i)$ for every $i \leq k$ and $r \leq r_1 + \cdots + r_k$. By definition of $\quasi{\tau_n}$ in Lemma~\ref{lem:squasitame}, we have $(r_1,\quasi{\tau_{n-1}}(T_1)),\dots,(r_k,\quasi{\tau_{n-1}}(T_k)) \in \quasi{\tau_n}(T)$. Moreover, since $r_1 + \cdots + r_k \in f_{n-1}(T_i)$ for every $i \leq k$ it is immediate by induction hypothesis that $r_1 + \cdots + r_k \in f_{n-1}(\quasi{\tau_{n-1}}(T_i))$ for every $i \leq k$. Altogether, we obtain that $r \in f_n(\quasi{\tau_n}(T))$, finishing the proof of the lemma.
\end{proof}

\begin{lemma} \label{lem:fmapmult}
  For every $n \in \nat$ and $T,T' \in Q_n$, we have $f_n(T) \cdot f_n(T') \subseteq  f_n(T \cdot T')$.
\end{lemma}

\begin{proof}
  We proceed by induction on $n \in \nat$. We first handle the case $n = 0$. Let $r \in f_0(T) \cdot f_0(T')$. We have $s \in f_0(T)$ and $s' \in f_0(T')$ such that $r = ss'$. By definition, this yields $r_1,\dots,r_k \in T$ and $r'_1,\dots,r'_{k'} \in T'$ such that $s \leq r_1 + \cdots + r_k$ and $s' \leq r'_1 + \cdots + r'_{k'}$. It follows that $ss' \leq \sum_{i \leq k} \sum_{j \leq k'} r_ir'_j$. Since $r_ir'_j \in T \cdot T'$ for every $i \leq k$ and $j \leq k'$, it follows that $ss' \in f_0(T \cdot T')$.

  Assume now that $n \geq 1$. Let $r \in f_n(T) \cdot f_n(T')$. We have $s \in f_n(T)$ and $s' \in f_n(T')$ such that $r = ss'$. By definition, this yields $(r_1,T_1),\dots,(r_k,T_k) \in T$ and $(r'_1,T'_1),\dots,(r'_{k'},T'_{k'}) \in T'$ such that,
  \begin{itemize}
  \item $s \leq r_1 + \cdots + r_k$ and $r_1 + \cdots + r_k \in f_{n-1}(T_i)$ for every $i \leq k$.
  \item $s' \leq r'_1 + \cdots + r'_{k'}$ and $r'_1 + \cdots + r'_{k'} \in f_{n-1}(T'_j)$ for every $j \leq k'$.
  \end{itemize}
  Clearly, we have $ss' \leq \sum_{i \leq k} \sum_{j \leq k'} r_ir'_j$. Moreover, for every $i \leq k$ and $j \leq k'$, we have,
  \[
    \sum_{i \leq k} \sum_{j \leq k'} r_ir'_j \in f_{n-1}(T_i) \cdot f_{n-1}(T'_j).
  \]
  By induction hypothesis, this yields,
  \[
    \sum_{i \leq k} \sum_{j \leq k'} r_ir'_j \in f_{n-1}(T_i \cdot T'_j).
  \]
  Finally, it is immediate that for every $i \leq k$ and $j \leq k'$, we have $(r_ir'_j,T_i \cdot T'_j) \in T \cdot T'$. Altogether, this yields $ss' \in f_n(T \cdot T')$ by definition.
\end{proof}

\medskip
\noindent
{\bf Proof of Proposition~\ref{prop:sound2}.} We now turn to the main argument. We fix $S \subseteq (\sclac) \times R$ which is \bpol{\Cs}-saturated (for $\rho$). We have to show that $S(D) \subseteq f_n(\tau_n(D))$ for every $n \in \nat$ and $D \in \sclac$. The argument is an induction on $n \in \nat$.

\smallskip
\noindent {\it Base case.} Assume that $n = 0$ and let $D \in \sclac$. We show that $S(D) \subseteq f_0(\tau_0(D))$. Let \mbox{$r \in S(D)$}, \emph{i.e.}, $(D,r) \in S$. Since $S$ is \bpol{\Cs}-saturated, we get from~\eqref{eq:bpol:sat} that there exist $r_1,\dots,r_k \in R$ such that $r \leq r_1 + \cdots + r_k$ and \mbox{$(D,r_i,\{r_1 + \cdots + r_k\}) \in \rbpols$} for every \mbox{$i \leq k$}.  For $i \leq k$, one may verify from the definition of \rbpols that since \mbox{$(D,r_i,\{r_1 + \cdots + r_k\}) \in \rbpols$}, we have $w_i \in A^*$ such that $\ctype{w_i} = D$ and $\rho(w_i) = r_i$. In particular, $w_1,\dots,w_k \in D$ and it is therefore immediate from the definition of $\tau_0$ that $r_1,\dots,r_k \in \tau_0(D)$. Since $r \leq r_1 + \cdots + r_k$, we obtain $r \in f_0(\tau_0(D))$ by definition of $f_0$, which concludes the proof.

\smallskip
\noindent
{\it Inductive step.} We now assume that $n \geq 1$. The argument is based on the following lemma, which is where we use the characterization of \pol{\Cs}-optimal \imprints (\emph{i.e.} Theorem~\ref{thm:half:mainpolc}): we apply it to the quasi-\mratm $\tau_{n-1}$. Moreover, this is also where we apply induction on $n$ in Proposition~\ref{prop:sound2}.

\begin{lemma} \label{lem:bpol:soundopti}
  For all $(D,q,U) \in \rbpols$, we have $T \in Q_{n-1}$ satisfying the following properties:
  \[
    \text{$T \in \opti{\pol{\Cs}}{\rho_*\inv(q) \cap D,\tau_{n-1}}$ and $U \subseteq f_{n-1}(T)$}.
  \]
\end{lemma}

We start by explaining how Lemma~\ref{lem:bpol:soundopti} can be applied to complete the main proof. Consider $D \in \sclac$. We show that $S(D) \subseteq f_n(\tau_n(D))$. Let $r \in S(D)$ (\emph{i.e.}, $(D,r) \in S$). We prove that $r \in f_n(\tau_n(D))$.

Since $S$ is \bpol{\Cs}-saturated, we get from~\eqref{eq:bpol:sat} that there exist $r_1,\dots,r_k \in R$ such that $r \leq r_1 + \cdots + r_k$ and $(D,r_i,\{r_1 + \cdots + r_k\}) \in \rbpols$ for every $i \leq k$. Lemma~\ref{lem:bpol:soundopti} then yields $T_i \in Q_{n-1}$ such that $T_i \in \opti{\pol{\Cs}}{\rho_*\inv(r_i) \cap D,\tau_{n-1}}$ and \mbox{$r_1 + \cdots + r_k \in f_{n-1}(T_i)$} for every $i \leq k$. By definition, $\tau_n$ is the \ratm $\lrataux{\pol{\Cs}}{\rho_*}{\tau_{n-1}}: 2^{A^*} \to 2^{R\times Q_{n-1}}$. Thus, if we unravel the definition, the hypothesis that $T_i \in \opti{\pol{\Cs}}{\rho_*\inv(r_i) \cap D,\tau_{n-1}}$ exactly says that,
\[
  (r_i,U_i) \in \tau_n(D).
\]
Since we also have $r \leq r_1 + \cdots + r_k$ and $r_1 + \cdots + r_k \in f_{n-1}(T_i)$ for every $i \leq k$, it is immediate by definition of $f_n$ that $r \in  f_n(\tau_n(D))$, finishing the proof of Proposition~\ref{prop:sound2}. It remains to prove Lemma~\ref{lem:bpol:soundopti}.

\newcommand{\poctau}{\popti{\pol{\Cs}}{\alpha}{\tau_{n-1}}}
\begin{proof}[Proof of Lemma~\ref{lem:bpol:soundopti}]
  We first apply Theorem~\ref{thm:half:mainpolc} and then use the result to prove the lemma. Clearly, the Cartesian product $(\sclac) \times R$ is a monoid when equipped with the componentwise multiplication. Let $\alpha: A^* \to (\sclac) \times R$ be the morphism defined by $\alpha(w) = (\ctype{w},\rho(w))$. Clearly, $\alpha$ is a \Cs-compatible morphism: for every pair $(D,r) \in (\sclac) \times R$, it suffices to define $\ctype{(D,r)} = D$. Consequently, since we also know that $\tau_{n-1}$ is quasi-\tame by Lemma~\ref{lem:squasitame}, we may apply Theorem~\ref{thm:half:mainpolc} to obtain a description of the set $\poctau \subseteq (\sclac) \times R \times Q_{n-1}$. Consider the least \pol{\Cs}-saturated subset $X$ of $(\sclac) \times R \times Q_{n-1}$ for $\alpha$ and $\tau_{n-1}$. Theorem~\ref{thm:half:mainpolc} yields that,
  \begin{equation} \label{eq:bpol:applipolcs}
    \poctau = \dclos_{Q_{n-1}} \big\{(D,q,\quasi{\tau_{n-1}}(P)) \mid (D,q,P) \in X\big\}.
  \end{equation}
  The proof of Lemma~\ref{lem:bpol:soundopti} is now based on the following lemma which we shall prove by induction on the construction of an element of \rbpols.

  \begin{lemma} \label{lem:bpol:subres}
    For every $(D,q,U) \in \rbpols$. There exists $P \in Q_{n-1}$ such that $(D,q,P) \in X$ and $U \subseteq f_{n-1}(P)$.
  \end{lemma}

  Let us first use Lemma~\ref{lem:bpol:subres} to conclude the proof of Lemma~\ref{lem:bpol:soundopti}. Let $(D,q,U) \in \rbpols$. We have to exhibit $T \in Q_{n-1}$ such that $(D,q,T) \in \poctau$ and $U \subseteq f_{n-1}(T)$.

  Using Lemma~\ref{lem:bpol:subres}, we get $P \in Q_{n-1}$ such that $(D,q,P) \in X$ and $U \subseteq f_{n-1}(P)$. In view of~\eqref{eq:bpol:applipolcs}, this yields,
  \[
    (D,q,\quasi{\tau_{n-1}}(P))  \in \poctau.
  \]
  By definition of the set \poctau and of the morphism $\alpha$, this exactly says that,
  \[
    \quasi{\tau_{n-1}}(P) \in \opti{\pol{\Cs}}{\rho_*\inv(q) \cap D,\tau_{n-1}}.
  \]
  Moreover, Lemma~\ref{lem:fmapmu} yields that $U \subseteq f_{n-1}(P) \subseteq f_{n-1}(\quasi{\tau_{n-1}}(P))$. Hence, the desired property holds for $T = \quasi{\tau_{n-1}}(P)$ and we are finished.

  \medskip

  It remains to prove Lemma~\ref{lem:bpol:subres}. Consider $(D,q,U) \in \rbpols$. By definition of \rbpols, we know that $(D,q,U) \in \rbpols$ is built from trivial elements using three operations: extended downset, multiplication and $S$-restricted closure. We proceed by induction on this construction. There are four cases depending on the last operation used to build $(D,q,U) \in \rbpols$.

  \medskip
  \noindent
  {\bf Base case: trivial elements.} In that case, there exists $w \in A^*$ such that $D = \ctype{w}$, $q = \rho(w)$ and $U = \{\rho(w)\}$. Since $X$ is \pol{\Cs}-saturated for $\rho_*$ and $\tau_{n-1}$, we know that $(\ctype{w},\rho(w),\tau_{n-1}(w)) \in X$. Hence, it remains to prove that $U = \{\rho(w)\} \subseteq f_{n-1}(\tau_{n-1}(w))$. It will then be immediate that Lemma~\ref{lem:bpol:subres} holds for $P = \tau_{n-1}(w)$. This is immediate from the following fact.

  \begin{fct} \label{fct:bpol:fmaptriv}
    For every $m \in \nat$, we have $\rho(w) \in f_m(\tau_m(w))$.
  \end{fct}

  \begin{proof}
    This is a simple induction on $m$. When $m = 0$, we have $\rho(w) \in \tau_0(w)$ by definition of $\tau_0$. It follows that $\rho(w) \in f_0(\tau_0(w))$ by definition of $f_0$. Assume now that $m \geq 1$. Recall that $\tau_m$ is $\lrataux{\pol{\Cs}}{\rho_*}{\tau_{m-1}}: 2^{A^*} \to 2^{R \times Q_{m-1}}$ by definition. Therefore,
    \[
      \tau_m(w)(\rho(w)) = \opti{\pol{\Cs}}{\rho_*\inv(\rho(w)) \cap \{w\},\tau_{m-1}} = \opti{\pol{\Cs}}{\{w\},\tau_{m-1}}
    \]
    Thus, $\tau_{m-1}(w) \in \tau_m(w)(\rho(w))$. Moreover induction yields $\rho(w) \in f_{m-1}(\tau_{m-1}(w))$. Therefore, by definition of $f_m$, we have $\rho(w) \in f_m(\tau_m(w))$.
  \end{proof}
\noindent
{\bf Inductive case~1: extended downset.} We have $(D,r,U') \in \rbpols$ and such that $U \subseteq \dclosr U'$. By induction, we obtain $P \in Q_{n-1}$ such that $(D,r,P) \in X$ and $U' \subseteq  f_{n-1}(P)$. Clearly, this implies $U \subseteq  \dclosr f_{n-1}(P)$. Moreover, $\dclosr f_{n-1}(P) = f_{n-1}(P)$ by definition. Therefore, $U \subseteq  f_{n-1}(P)$ which concludes this case.

 \medskip
  \noindent
  {\bf Inductive case~2: multiplication.} We have $(D_1,r_1,U_1) \in \rbpols$ and $(D_2,r_2,U_2) \in \rbpols$ such that $D = D_1 \cmult D_2$, $q = r_1r_2$ and $U = U_1U_2$. By induction, we obtain $P_1,P_2 \in Q_{n-1}$ such that $(D_1,r_1,P_1),(D_2,r_2,P_2) \in X$, $U_1 \subseteq  f_{n-1}(P_1)$ and $U_2 \subseteq  f_{n-1}(P_2)$. Since $X$ is \pol{\Cs}-saturated, it is closed under multiplication and we get,
  \[
    (D,q,P_1P_2) =  (D_1 \cmult D_2,r_1r_2,P_1P_2) \in X
  \]
  Moreover, Lemma~\ref{lem:fmapmult} yields that,
  \[
    U = U_1U_2 \subseteq  f_{n-1}(P_1) \cdot f_{n-1}(P_2) \subseteq  f_{n-1}(P_1P_2)
  \]
  Altogether, it follows that Lemma~\ref{lem:bpol:subres} holds for $P = P_1P_2$.

  \medskip
  \noindent
  {\bf Inductive case~3: $S$-restricted closure.} In that case, we have idempotents $(E,f,F) \in \rbpols$ such that $D = E$, $q = f$ and $U = F \cdot S(E) \cdot F$.

  By induction, we get $V \in Q_{n-1}$ such that $(E,f,V) \in X$ and $F \subseteq f_{n-1}(V)$. Since $Q_{n-1}$ is a finite monoid, there exists a number $n \geq 1$ such that $V^n$ is a multiplicative idempotent of $Q_{n-1}$. Since $(E,f,V) \in X$, and $X$ is closed under multiplication (as it is \pol{\Cs}-saturated), we obtain $(E,f,V^n) \in X$. We may again use the hypothesis that $X$ is \pol{\Cs}-saturated to apply \pol{\Cs}-closure and obtain,
  \[
    (E,f,V^n \cdot \tau_{n-1}(E) \cdot V^n) \in X.
  \]
  Since $D = E$ and $q = f$, it now remains to show that $U = F \cdot S(E) \cdot F \subseteq f_{n-1}(V^n \cdot \tau_{n-1}(E) \cdot V^n)$. It will then be immediate that Lemma~\ref{lem:bpol:subres} holds for $P = V^n \cdot \tau_{n-1}(E) \cdot V^n$. It is immediate by induction in Proposition~\ref{prop:sound2} that,
  \[
    S(E) \subseteq f_{n-1}(\tau_{n-1}(E)).
  \]
  Moreover, since we already know that $F \subseteq f_{n-1}(V)$ is an idempotent by definition, it follows from Lemma~\ref{lem:fmapmult} that,
  \[
    F \cdot S(E) \cdot F =  F^n \cdot S(E) \cdot F^n \subseteq f_{n-1}(V^n \cdot \tau_{n-1}(E) \cdot V^n).
  \]
  This concludes the proof of Lemma~\ref{lem:bpol:subres}.
\end{proof}

\section{Completeness in Theorem~\ref{thm:bpol:carac}}
\label{sec:comp}
In this section, we prove the completeness direction in Theorem~\ref{thm:bpol:carac}. As for the soundness proof, the section is divided in two parts. We start with a preliminary result which applies to Boolean closure in general, \emph{i.e.}, to classes of the form $\bool{\Ds}$ when \Ds is an arbitrary lattice. We then use it in the special case $\Ds = \pol{\Cs}$ to prove the completeness direction of Theorem~\ref{thm:bpol:carac}.

\subsection{Preliminary result}

We fix a lattice \Ds. Moreover, we let $\rho: 2^{A^*} \to R$ be a \textbf{\nice} \ratm. Given a language $L \in \Ds$, we characterize the set $\opti{\bool{\Ds}}{L,\rho} \subseteq R$ using \Ds-optimal \bratauxbd-\imprints when  $\bratauxbd: 2^{A^*} \to 2^R$ is the \ratm introduced in Proposition~\ref{prop:areratms}.

\begin{proposition} \label{prop:comp1}
  Let $L \in \Ds$. For all $s \in \opti{\bpol{\Cs}}{L,\rho}$, we have $r_1,\dots,r_k \in R$ such that $s \leq r_1 + \cdots + r_k$ and for every $i \leq k$, we have $\{r_1+ \cdots +r_k\} \in \opti{\Ds}{L \cap \rho_*\inv(r_i),\bratauxbd}$.
\end{proposition}

\begin{proof}
  We fix $L \in \Ds$ and $s \in \opti{\bpol{\Cs}}{L,\rho}$ for the proof.  For every $q \in R$, we let $\Hb_q$ be an optimal \Ds-cover of $L \cap \rho_*\inv(q)$ for $\tau$. We have the following fact.

  \begin{fact} \label{fct:bpol:thepart}
    There exists a \bool{\Ds}-cover \Kb of $L$ such that for every $K \in \Kb$ and  every $q \in R$ such that $K \cap \rho_*\inv(q)\neq\emptyset$, there exists $H \in \Hb_q$ such that $K \subseteq H$.
  \end{fact}

  \begin{proof}
    Let $\Hb = \bigcup_{q \in R} \Hb_q$ and consider the following equivalence $\sim$ defined on $L$. For every $u,v \in L$, we let $u \sim v$ when $u \in H \Leftrightarrow v \in H$ for every $H \in \Hb$. Let \Kb be the partition of $L$ into $\sim$-classes. By definition, \Kb is a cover of $L$. Moreover, it is a \bool{\Ds}-cover. Indeed, by definition, \Kb only contains Boolean combinations of $L \in \bool{\Ds}$ with languages in \Hb (which are in $\Ds\subseteq\bool{\Ds}$). It remains to show that \Kb satisfies the property of the fact.

    Let $q \in R$ and assume that there exists $w \in K \cap \rho_*\inv(q)$. By definition of \Kb, we have $K \subseteq L$ which means that $w \in L \cap \rho_*\inv(q)$. Therefore, since $\Hb_q$ is a cover of $L \cap \rho_*\inv(q)$ by definition, we have $H \in \Hb_q$ such that $w \in H$. Consequently, $K \cap H \neq \emptyset$. Finally, since $K$ is a $\sim$-class by definition of \Kb, it follows from the definition of $\sim$ that $K \subseteq H$.
  \end{proof}

  We let \Kb be the \bool{\Ds}-cover of $L$ described in Fact~\ref{fct:bpol:thepart}. By definition, we have $L \subseteq \bigcup_{K \in \Kb} K$. Hence, it follows from Lemma~\ref{fct:lattice:optunion} that,
  \[
  \opti{\bool{\Ds}}{L,\rho} \subseteq \bigcup_{K \in \Kb} \opti{\bool{\Ds}}{K,\rho}
  \]
  Consequently, since $s \in \opti{\bool{\Ds}}{L,\rho}$, we get some $K \in \Kb$ such that $s \in \opti{\bool{\Ds}}{K,\rho}$. We let \Vb be an optimal \bool{\Ds}-cover of $K$ for $\rho$. Since $K \in \bool{\Ds}$, we may choose \Vb such that $V \subseteq K$ for all $V \in \Vb$. By definition of \Vb, we have $\prin{\rho}{\Vb} = \opti{\bool{\Ds}}{K,\rho}$ and we obtain that $s \in \prin{\rho}{\Vb}$. Therefore, there exists $V \in \Vb$ such that $s \leq \rho(V)$.

  Since $\rho$ is \nice by hypothesis, we have $w_1,\dots,w_k \in V$ such that $\rho(V) = \rho(w_1) + \cdots + \rho(w_k)$. We let $r_i = \rho(w_i)$ for every $i \leq k$. by definition, we have $s \leq r_1 + \cdots + r_k$. Therefore, it remains to show that $\{r_1+ \cdots +r_k\} \in \opti{\Ds}{L \cap \rho_*\inv(r_i),\brataux{\bool{\Ds}}{\rho}}$ for every $i \leq k$.

  We fix $i \leq k$ for the proof. By definition, $\rho(w_i) = r_i$ and $w_i \in V \subseteq K$. Hence, $w_i \in K \cap \rho_*\inv(r_i)$ and it follows from the definition of \Kb in Fact~\ref{fct:bpol:thepart} that there exists a language $H \in \Hb_{r_i}$ such that $K \subseteq H$. Recall that $r_1 + \cdots + r_k = \rho(V)$. Moreover, $\rho(V) \in \prin{\rho}{\Vb}$ and we have $\prin{\rho}{\Vb} = \opti{\bool{\Ds}}{K,\rho}$. Consequently, by Fact~\ref{fct:inclus2},
  \[
  r_1 + \cdots + r_k \in \opti{\bool{\Ds}}{K,\rho} \subseteq \opti{\bool{\Ds}}{H,\rho}
  \]
  By definition of the \ratm \brataux{\bool{\Ds}}{\rho}, we have $\opti{\bool{\Ds}}{H,\rho} = \brataux{\bool{\Ds}}{\rho}(H)$. Thus, the above can be rephrased as follows:
  \[
  \{r_1+\cdots+r_k\} \subseteq \brataux{\bool{\Ds}}{\rho}(H)
  \]
  Consequently, since $H \in \Hb_{r_i}$, we have $\{r_1+\cdots+r_k\} \in \prin{\brataux{\bool{\Ds}}{\rho}}{\Hb_{r_i}}$. Recall that we defined $\Hb_{r_i}$ as an optimal \Ds-cover of $L \cap \rho_*\inv(r_i)$ for $\brataux{\bool{\Ds}}{\rho}$. Hence, we obtain,
  \[
  \{r_1 + \cdots + r_k\} \in \opti{\Ds}{L \cap \rho_*\inv(r_i),\brataux{\bool{\Ds}}{\rho}}
  \]
  This concludes the proof of Proposition~\ref{prop:comp1}.
\end{proof}

\subsection{Completeness proof for Theorem~\ref{thm:bpol:carac}}

We may now come back to our main objective: completeness in Theorem~\ref{thm:bpol:carac}. We fix a finite \vari \Cs and a \nice \mratm $\rho: 2^{A^*} \to R$. We prove that \cbpopti is \bpol{\Cs}-saturated for~$\rho$. Since we already showed in the previous section that every \bpol{\Cs}-saturated subset is included in \cbpopti, Theorem~\ref{thm:bpol:carac} will follow: \cbpopti is the greatest \bpol{\Cs}-saturated subset of $R$. For the sake of avoiding clutter, we write~$S$ for $\cbpopti$.

\begin{remark}
  Contrary to the soundness direction, we do need the hypothesis that $\rho$ is \nice. This is required for applying the above preliminary result, Proposition~\ref{prop:comp1}.
\end{remark}

By Theorem~\ref{thm:polclos}, \pol{\Cs} is a lattice. Therefore, we may instantiate Proposition~\ref{prop:comp1} for our \nice \mratm $\rho: 2^{A^*} \to R$ in the special case when $\Ds = \pol{\Cs}$. We complete it with another result specific to this special case. As for the soundness direction, its proof is based on Theorem~\ref{thm:half:mainpolc}: the characterization of \pol{\Cs}-optimal \imprints. Together, these two results imply that $S = \cbpopti$ is \bpol{\Cs}-saturated for $\rho$.

\begin{proposition} \label{prop:comp2}
  Let $D \in \sclac$ and $r \in R$. For all $U \in \opti{\pol{\Cs}}{D \cap \rho_*\inv(r),\bratauxbc}$, we have $(D,r,U) \in \rbpol{S}$.
\end{proposition}

We first combine Proposition~\ref{prop:comp1} and Proposition~\ref{prop:comp2} to prove that $S$ is \bpol{\Cs}-saturated for $\rho$ and finish the completeness proof. We have to show that~\eqref{eq:bpol:sat} is satisfied. That is, given $(D,r) \in S$, we have to exhibit $r_1,\dots,r_k \in R$ such that $r \leq r_1 + \cdots + r_k$ and $(D,r_i,\{r_1+ \cdots +r_k\}) \in \rbpol{S}$ for every $i\leq k$.

Since $D \in \sclac$, we have $D \in \Cs \subseteq \pol{\Cs}$. Moreover, since  $S = \cbpopti$ and $(D,r) \in S$, we have $r \in \opti{\bpol{\Cs}}{D,\rho}$. Thus,  Proposition~\ref{prop:comp1} yields \mbox{$r_1,\dots,r_k \in R$} such that $r \leq r_1 + \cdots + r_k$ and $\{r_1+ \cdots +r_k\} \in \opti{\pol{\Cs}}{D \cap \rho_*\inv(r_i),\tau}$ for every $i \leq k$. Then, we obtain from Proposition~\ref{prop:comp2} that $(D,r,\{r_1+ \cdots +r_k\}) \in \rbpol{S}$ for every $i \leq k$. Altogether, this is exactly the property stated in~\eqref{eq:bpol:sat}. This concludes the completeness proof for Theorem~\ref{thm:bpol:carac}. It remains to prove Proposition~\ref{prop:comp2}.

\newcommand{\pocbrat}{\popti{\pol{\Cs}}{\alpha}{\tau}}
\begin{proof}[Proof of Proposition~\ref{prop:comp2}]
  The argument is based on Theorem~\ref{thm:half:mainpolc}. Clearly, the Cartesian product $(\sclac) \times R$ is a monoid when equipped with the componentwise multiplication. Let $\alpha: A^* \to (\sclac) \times R$ be the morphism defined by $\alpha(w) = (\ctype{w},\rho(w))$. Clearly, $\alpha$ is a \Cs-compatible morphism: for every pair $(D,r) \in (\sclac) \times R$, it suffices to define $\ctype{(D,r)} = D$.

  Moreover, \pol{\Cs} is a \pvari closed under concatenation and $\rho: 2^{A^*} \to R$ is a \mratm. Thus, Lemma~\ref{lem:nest:bpolrat} yields a \ratm $\tau: 2^{A^*} \to 2^R$ which is quasi-\tame for the endomorphism $\quasit: U\mapsto\dclosr U$ of $(2^R,\cup)$ and such that for every $K \in \pol{\Cs}$, we have $\tau(K) = \bratauxbc(K)$. We have the following key lemma which we need to apply Theorem~\ref{thm:half:mainpolc}.

  \begin{lemma} \label{lem:bpol:thecomp}
    The set $\rbpols \subseteq (\sclac) \times R \times 2^R$ is \pol{\Cs}-saturated for $\alpha$ and $\tau$.
  \end{lemma}

  Let us first use the lemma to conclude the proof of Proposition~\ref{prop:comp2}. Let $D \in \sclac$ and $r \in R$. Consider $U \in \opti{\pol{\Cs}}{D \cap \rho_*\inv(r),\bratauxbc}$. We show that $(D,r,U) \in \rbpol{S}$.  Since $\tau(K) = \bratauxbc(K)$ for every $K \in \pol{\Cs}$, it is immediate that,
  \[
  \opti{\pol{\Cs}}{D \cap \rho_*\inv(r),\tau} = \opti{\pol{\Cs}}{D \cap \rho_*\inv(r),\bratauxbc}
  \]
  Thus $U \in \opti{\pol{\Cs}}{D \cap \rho_*\inv(r),\tau}$ which means that $(D,r,U) \in \pocbrat$ by definition. Moreover, \rbpols is \pol{\Cs}-saturated for $\alpha$ and $\tau$ by Lemma~\ref{lem:bpol:thecomp}. Hence, since $\tau: 2^{A^*} \to 2^R$ is a quasi-\mratm for $\quasit: U \mapsto\dclosr U$, it follows from Theorem~\ref{thm:half:mainpolc} that,
  \[
  \pocbrat \subseteq \dclosp{2^R} \big\{(D',r',\dclosr V) \mid (D',r',V) \in \rbpols\big\}.
  \]
  Altogether, we get $V \in 2^R$ such that $(D,r,V) \in \rbpols$ and $U \subseteq \dclosr V$.  By closure under extended downset in the definition of \rbpols, this yields $(D,r,U) \in \rbpol{S}$ as desired.

  \medskip

  It remains to prove Lemma~\ref{lem:bpol:thecomp}: \rbpols is \pol{\Cs}-saturated for $\alpha$ and $\tau$. We have four conditions to check. We leave the trivial elements for last as we need the other properties to handle them.  It is immediate that $\rbpols$ is closed under downset and multiplication as stated in the definition of \pol{\Cs}-saturated subsets. This is implied by the closure under extended downset and multiplication in the definition of \rbpols. We turn to \pol{\Cs}-closure. Let $(E,f,F) \in \rbpols$ be a triple of multiplicative idempotents. We prove that $(E,f,F\cdot \tau(E) \cdot F) \in \rbpols$. We may use $S$-restricted closure in the definition of \rbpols to obtain $(E,f,F \cdot S(E) \cdot F) \in \rbpols$. Moreover, $S = \cbpopti$ by definition and $\tau(E)= \bratauxbc(E) = \opti{\bpol{\Cs}}{E,\rho}$ since $E\in\Cs\subseteq\pol{\Cs}$. Thus, $S(E) = \tau(E)$. We obtain $(E,f,F \cdot \tau(E) \cdot F) \in \rbpols$ as desired.

    We turn to the trivial elements. For $w \in A^*$, we show that \mbox{$(\ctype{w},\rho(w),\tau(w))\in \rbpols$}. First, we consider the special case $w = \veps$. By definition of \rbpols, we have  $(\ctype{\veps},\rho(\veps),\{\rho(\veps)\}) \in \rbpols$. Since this is a triple of multiplicative idempotents and we already established that \rbpols satisfies \pol{\Cs}-closure, this implies $(\ctype{\veps},\rho(\veps),\tau(\ctype{\veps})) \in \rbpols$. It then follows from closure under downset that $(\ctype{\veps},\rho(\veps),\tau(\veps)) \in \rbpols$ since $\tau(\veps) \subseteq \tau(\ctype{\veps})$. We now consider the case when $w = a$ for some $a\in A$. We already established that $(\ctype{\veps},\rho(\veps),\tau(\ctype{\veps})) \in \rbpols$. Moreover, we know that $(\ctype{a},\rho(a),\{\rho(a)\}) \in \rbpols$ by definition of \rbpols. Using closure under multiplication and extended downset, we get,
   \[
    (\ctype{a},\rho(a),\dclosr \left(\tau(\ctype{\veps}) \cdot \{\rho(a)\} \cdot \tau(\ctype{\veps})\right)) \in \rbpols
    \]
    We know that \pol{\Cs} is a \pvari of regular languages closed under marked concatenation by Theorem~\ref{thm:polclos}. Therefore, since $\tau$ and \bratauxbc coincide over languages in \pol{\Cs}, we obtain from Lemma~\ref{lem:nest:bpolm} that,
 \[
 \tau(\ctype{\veps}a\ctype{\veps}) = \dclosr \left(\tau(\ctype{\veps}) \cdot \{\rho(a)\} \cdot \tau(\ctype{\veps})\right)
 \]
 Altogether, we obtain $(\ctype{a},\rho(a),\tau(\ctype{\veps}a\ctype{\veps})) \in \rbpols$. Moreover, we have $\tau(a)\subseteq\tau(\ctype{\veps}a\ctype{\veps})$. Thus, we obtain from closure under downset that $(\ctype{a},\rho(a),\tau(a))\in \rbpols$ as desired. Finally, assume that $w = a_1 \cdots a_n$ for $n \geq 2$ and $a_1,\dots,a_n \in A$. We already established that $(\ctype{a_i},\rho(a_i),\tau(a_i))\in \rbpols$ for every $i \leq n$. Therefore, closure under multiplication and extended downset yield,
 \[
 (\ctype{w},\rho(w),\dclosr \left(\tau(a_1) \cdots \tau(a_n)\right)) \in \rbpols
 \]
 Since $\tau$ is quasi-\tame for the endomorphism $\quasit: U\mapsto\dclosr U$ of $(2^R,\cup)$, this yields $(\ctype{w},\rho(w), \tau(w)) \in \rbpols$, concluding the proof.
\end{proof}

\section{Conclusion}
\label{sec:conc}
We proved that separation and covering are decidable for all classes of the form \bpol{\Cs} when \Cs is a finite \vari. This yields separation and covering algorithms for a whole family of classes. Arguably, the most important one is the level two in the Straubing-Thérien hierarchy (which corresponds to the logic \bswd). Additionally, this result can be lifted to depth-two using an effective reduction of~\cite{pzsucc2} to the level two in the Straubing-Thérien hierarchy.

An interesting consequence of our results is that since we proved the decidability of separation for the level two in the Straubing-Thérien hierarchy, the main theorem of~\cite{pzqalt} is an immediate corollary: membership for this level is decidable. However, the algorithm of~\cite{pzqalt} was actually based on a characterization theorem: languages of level two in the Straubing-Thérien hierarchy are characterized by a syntactic property of a canonical recognizer (\emph{i.e.}, their syntactic monoid). It turns out that one can also deduce this characterization theorem from our results (this does require some combinatorial work however). In fact, one may generalize it to all classes \bpol{\Cs} when \Cs is a finite \vari.

Finally, the main and most natural follow-up question is much harder: can our results be pushed to higher levels within concatenation hierarchies? For now, we know that given any finite \vari \Cs, \pol{\Cs}, \bpol{\Cs} and \pol{\bpol{\Cs}} have decidable covering (the former and the latter are results of~\cite{pseps3j}). Consequently, the next relevant levels are \bpol{\bpol{\Cs}} and \pol{\bpol{\bpol{\Cs}}}.

\bibliographystyle{alpha}
\bibliography{main}

\end{document}